\documentclass[conference]{IEEEtran}

\usepackage{color}
\usepackage{graphicx}
\usepackage{float} 
\usepackage{subfigure}
\usepackage{algorithm2e}
\usepackage{eqnarray,amsmath}
\usepackage{etoolbox}
\usepackage[perpage]{footmisc}
\usepackage{lipsum}
\usepackage{multirow}

\usepackage{amsthm}
\usepackage{amssymb}
\newtheorem{theorem}{Theorem}[section]
\newtheorem*{theorem*}{Theorem}

\usepackage{amsmath}
\usepackage{array}
\usepackage{booktabs}
\usepackage{amssymb}
\usepackage{threeparttable}
\usepackage{graphicx}

\newcommand{\ie}{{\it i.e.}}

\makeatletter
\newcommand\blfootnote[1]{%
  \begingroup
  \renewcommand{\@makefntext}[1]{\noindent\makebox[1.8em][r]#1}
  \renewcommand\thefootnote{}\footnote{#1}%
  \addtocounter{footnote}{-1}%
  \endgroup
}
\makeatother

\begin{document}

\title{Secure and Efficient General Matrix Multiplication On Cloud Using Homomorphic Encryption}

\author{
\IEEEauthorblockN{1\textsuperscript{st} Yang Gao}
\IEEEauthorblockA{University of Central Florida\\
yang.gao@ucf.edu}\\
\IEEEauthorblockN{4\textsuperscript{th} Wujie Wen}
\IEEEauthorblockA{Lehigh University\\
wuw219@lehigh.edu}
\and
\IEEEauthorblockN{2\textsuperscript{nd} Gang Quan}
\IEEEauthorblockA{Florida International University\\
gaquan@fiu.edu}\\
\IEEEauthorblockN{5\textsuperscript{th} Liqiang Wang}
\IEEEauthorblockA{University of Central Florida\\
liqiang.wang@ucf.edu}
\and
\IEEEauthorblockN{3\textsuperscript{rd} Soamar Homsi}
\IEEEauthorblockA{Air Force Research Laboratory\\
soamar.homsi@us.af.mil}\\
}

\thispagestyle{plain}
\pagestyle{plain}
\RestyleAlgo{ruled}

\maketitle

\begin{abstract}
Despite the enormous technical and financial advantages of cloud computing, security and privacy have always been the primary concerns for adopting cloud computing facilities, especially for government agencies and commercial sectors with high-security requirements. Homomorphic Encryption (HE) has recently emerged as an effective tool in ensuring privacy and security for sensitive applications by allowing computing on encrypted data. One major obstacle to employing HE-based computation, however, is its excessive computational cost, which can be orders of magnitude higher than its counterpart based on the plaintext. In this paper, we study the problem of how to reduce the HE-based computational cost for general Matrix Multiplication (MM), {\it i.e.}, a fundamental building block for numerous practical applications, by taking advantage of the Single Instruction Multiple Data (SIMD) operations supported by HE schemes. Specifically, we develop a novel element-wise algorithm for general matrix multiplication, based on which we propose two \underline{HE}-based \underline{G}eneral \underline{M}atrix \underline{M}ultiplication (HEGMM) algorithms to reduce the HE computation cost. Our experimental results show that our algorithms can significantly outperform the state-of-the-art approaches of HE-based matrix multiplication.
\end{abstract}


\vspace{-0.2cm}

\section{Introduction}
Cloud computing has become an attractive solution for industry and individuals due to its flexibility, scalability, reliability, sustainability, and affordability \cite{varghese2018next,vasiljeva2017cloud}.
Despite the tremendous technical and business advantages of cloud computing, security has been one of the primary concerns for cloud users, especially for those with high-security requirements~\cite{scale2015state,rajaraman2014cloud}. 
Even though cloud platforms allow their users to have full control over security settings and policies, public cloud infrastructures are commonly shared among different users and applications, making the applications vulnerable to malicious attacks. 

\blfootnote{Source Code: https://github.com/EchizenG/HEGMM}
\blfootnote{Approved for Public Release on 06 Mar 2024. Distribution is Unlimited. Case Number: 2024-0184  (original case number(s):  AFRL-2024-0944)}

\vspace{-0.3cm}
Homomorphic Encryption (HE)~\cite{rivest1978data,gentry2009fully,brakerski2014BGV} has emerged as an effective tool to address the security and privacy concerns associated with outsourcing data and computation to untrusted third parties, such as public cloud service providers. HE maintains data secrecy while in transit and during processing and assures that the decrypted results are identical to the outcome when the same operations are applied to the data in plaintext. HE has raised growing interest from researchers and practitioners of many security- and privacy-sensitive cloud applications in various domains such as health care, finance, and government agencies. One of the grand challenges, however, is how to deal with the tremendous computational cost for HE computations, which can be orders of magnitude higher than that in the plaintext space~\cite{ran2022cryptogcn}. Unless HE computation cost can be effectively reduced, it would be infeasible to apply HE schemes in practical cloud applications.   

In this paper, we study the problem of how to reduce HE computation cost for general matrix multiplications (MM) by taking advantage of the single instruction multiple data (SIMD) scheme for HE operations~\cite{smart2014fully}. The SIMD scheme enables multiple data values to be packed into one ciphertext, and one single HE operation can be performed on all data elements in the ciphertext simultaneously. Accordingly, we develop a novel approach for HE-based MM operations, focusing on source matrices of arbitrary shapes. Specifically, we make the following contributions. 
\vspace{-0.2cm}
\begin{enumerate}
\item We present a novel element-wise method for MM. This method is general and can be applied to source matrices of arbitrary shapes with  a significant performance improvement;  
\item We develop two HE MM algorithms, with the second one improving the first one significantly.
Our HE MM algorithms pack matrix elements judiciously in encrypted message ``slots'' and perform pertinent operations by taking advantage of the SIMD structure in HE schemes to reduce the number of primitive HE operations, such as HE multiplications, rotations, and additions, which are computationally expensive, and therefore can significantly reduce the computational cost; 
\item We perform a rigorous analysis for the logical correctness of the algorithms and their complexities;   
\item We implement our algorithms using a Python HE library, Pyfhel~\cite{ibarrondo2021pyfhel}. Extensive experimental results show that our proposed algorithms can significantly outperform the state-of-the-art approaches.
\end{enumerate}

\section{Background and Related Work}
In this section, we briefly introduce the relevant background of HE and discuss the related work. 
\begin{table*}[htbp]
  \centering
  \caption{Comparison of Computational cost for HE vs. Plaintext Operations.}
    \begin{tabular}{c|ccccccc}
    \textbf{Operations} & \multicolumn{1}{l}{\textbf{Encryption (ms)}} &  \multicolumn{1}{l}{\textbf{decryption (ms)}} & \multicolumn{1}{l}{\textbf{Message Size (MB)}} & \multicolumn{1}{l}{\textbf{Addition (ms)}} & \multicolumn{2}{c}{\textbf{Multiplication (ms)}$^\dagger$} & \multicolumn{1}{l}{\textbf{Rotation (ms)}}\\
    \hline
    \textbf{HE} & 5.50 & 2.57 & 0.5 & 0.550 & 20.874(CC) & 4.138(CP) & 5.350 \\
    \textbf{Plaintext}  & - & - & - & 0.009 &  0.035 & 0.035 & 0.130 \\
    \hline
    \textbf{Ratio} & - & - & - & 61.1 & 596.4(CC) & 118.23(CP) & 41.15\\
    \end{tabular}%
  \label{tab:cost}%
  \begin{tablenotes}
    \footnotesize
    \item[\dagger] $\dagger$ \emph{CC} is the Multiplication between two ciphertexts while \emph{CP} is the Multiplication between ciphertext and plaintext.
  \end{tablenotes}
\end{table*}

\subsection{Homomorphic Encryption (HE)}

Homomorphic encryption (e.g. BGV~\cite{brakerski2014BGV}, BFV~\cite{fan2012BFV,Bfv12}, and CKKS~\cite{cheon2017CKKS}) enables computations to be performed based on encrypted data, with results still in encrypted form. As such, HE represents a promising tool to greatly enhance data privacy and security, especially when outsourcing computations to the public cloud. In the meantime, HE can be extremely computationally intensive~\cite{ames2015secure}, and improving its computation efficiency is key to making this technology practical for real applications. 

When performing HE matrix multiplication on the cloud, source matrices are first encrypted by clients and transferred to the cloud, and the results are transferred back to clients for decryption. Encrypting each individual element of a matrix into one cyphertext can lead to excessive encryption, decryption, and communication costs, in addition to a large number of HE operations. Table~\ref{tab:cost} shows our profiling results on encryption/decryption latency, message size, and computational costs with different HE operations (More detailed experimental settings are discussed in Sec.~\ref{Exp-setup}).  

To this end, Gentry and Halevi \cite{smart2014fully} proposed an efficient key generation technique that enables SIMD operations in HE. By encrypting multiple data items into one ciphertext, one single operation can be applied to all encrypted elements in the same ciphertext simultaneously, and thus, space and computing resources can be used more efficiently. 

As an example, BFV \cite{fan2012BFV,Bfv12} can support a number of primitive HE operations such as HE-Add (Addition), HE-Mult (Multiplication), HE-CMult (Constant Multiplication), and HE-Rot (Rotation). Given ciphertexts $ct_x=Enc(x_0,x_1,...,x_n)$, $ct_y=Enc(y_0,y_1,...,y_n)$ and a plaintext $pt=(p_0,p_1,...,p_n)$, we have

\begin{itemize}
    \item HE-Add: $ct_x+ct_y=Enc(x_0+y_0,x_1+y_1,...,x_n+y_n)$
    \item HE-Mult: $ct_x\times ct_y=Enc(x_0\times y_0,x_1\times y_1,...,x_n\times y_n)$
    \item HE-CMult: $ct_x\times pt=Enc(x_0\times p_0,x_1\times p_1,...,x_n\times p_n)$
    \item HE-Rot: $Rot(ct_x, i)= Enc(x_i, x_{i+1},.., x_n, x_0,.., x_{i-1})$
\end{itemize}

The HE operations are computationally intensive and can consume 
excessive computational time. In addition, HE operations also introduce \emph{noises} when performed on encrypted data~\cite{nocker2023he}, which must be well under control for the results to be decrypted successfully. Several HE operations, especially HE-Mult, can be extremely time-consuming (approximately 600$\times$ higher than its counterpart as shown in Table~\ref{tab:cost}) and introduce much larger noise~\cite{reagen2021cheetah}. Therefore, reducing the number of HE operations (especially the HE-Mult operations) becomes critical in designing practical applications, such as matrix multiplication, under the HE framework. 

\subsection{Related Work}

There are numerous research efforts on improving the computational efficiency of MM (e.g.~\cite{masliah2019algorithms, nagasaka2018high,jiang2020novel, liu2014efficient, nagasaka2018high, 8374488, zhang2020sparch}). However, none of them can be readily adapted to optimize the computation efficiency of MM in the context of HE computation. 

A naive method for HE MM is to encrypt each row/column in each matrix and then compute it using the traditional MM method. For the HE MM of $\mathcal{A}_{m\times l} \times \mathcal{B}_{l\times n}$, this would result in excessive storage requirements and computation times: $m\times n$ encrypted messages and totally $m\times l \times n$ HE-Mult operations. Another simple and intuitive approach (e.g.~\cite{lu2016using}) is to transform the MM problem into the matrix-vector multiplication problem and then adopt the SIMD scheme~\cite{halevi2014algorithms,smart2014fully} to perform the calculation. However, this requires $m+n$ ciphertexts and $m \times n$ homomorphic multiplication operations, which are still very costly.  

Duong et al.~\cite{Duong2017TMMP} and Mishra et al.~\cite{Mishra2017ISPE} presented approaches to pack the source matrix into a single polynomial, and then perform HE MM based on secure computation of inner product of encrypted vectors. It works for one single HE MM with well-defined dimensions but becomes problematic when multiple successive HE MMs are required in a cloud center. 

Jiang et al. \cite{jiang2018secure} proposed an intriguing HE MM approach for \emph{square matrix} with $O(d)$ computational complexity. They expanded their HE MM algorithm to handle rectangle MM ($\mathcal{A}_{l \times d} \times \mathcal{B}_{d\times d}$) with $l \le d$ and $d\mod l=0$. Source matrices can be enlarged to suit shape requirements for MM with variable shapes, although this may increase processing time and resource utilization. Huang et al.~\cite{Huang2023TDSC} advocated using blocking to better handle rectangular MM with source matrices as \emph{block} matrices with square matrices. This method is appealing for big matrices that cannot fit in one ciphertext. However, it is limited to square or two source rectangular matrices with integer multiple columns and rows.

Rathee et al. \cite{rathee2018faster} proposed to encrypt source matrices into the two-dimensional \emph{hypercube} structure~\cite{halevi2014algorithms} and then transform the MM problem to a series of matrix-vector multiplication problems. Huang et al. \cite{huang2022secure} extended this approach to make it applicable to general MM. 
As shown in section~\ref{subsec:HEGMM-Enhanced}, we have developed a more effective algorithm with higher computational efficiency. 

\section{Approaches}
When performing HE matrix multiplication in the SIMD manner, we need to make sure that two operands are aligned and located at the same location, {\it i.e.}, the same slot in the two encoded ciphertexts. Rearranging individual slots in an encrypted message can be costly. Therefore, a key to the success of reducing the computational complexity of the HE MM is how to perform the MM using element-by-element additions and multiplication operations. In what follows, we first introduce a novel algorithm to calculate MM with arbitrary dimensions using element-wise additions and multiplications. We then discuss in more detail how we implement the HE MM algorithm on packed ciphertexts in the SIMD manner and its enhanced version.

\subsection{The Matrix Multiplication Method using Element-Wise Computations}
\label{subsec:Element-WiseMM}

Consider an MM problem, $\mathcal{C}_{m\times n} = \mathcal{A}_{m\times l} \times \mathcal{B}_{l\times n}$, where $m, l,$ and $n \in \mathbb{Z}^+$. Our goal is to develop an algorithm such that $\mathcal{C}_{m\times n} = \sum_i \mathcal{A}_i \odot \mathcal{B}_i$, where $\mathcal{A}_i$, $\mathcal{B}_i$ are certain transformations of $\mathcal{A}_{m\times l}$, $\mathcal{B}_{l\times n}$, respectively, and $\odot$ represents the element-wise multiplication. For ease of our presentation, we define four matrix transformation operators as follows:    
\begin{eqnarray}
  \sigma(\mathcal{A})_{i,j} &=& \mathcal{A}_{i,[i+j]_l}, \quad  0 \leq i <m, 0 \leq j <l \label{def:sigma} \\ 
  \tau(\mathcal{B})_{i,j} &=& \mathcal{B}_{[i+j]_l,j}, \quad  0 \leq i <l, 0 \leq j <n \label{def:tau} \\ 
  \epsilon^k_{m\times n}(\mathcal{A})_{i,j} &=& \mathcal{A}_{i,[j+k]_l}, \quad  0 \leq i <m, 0 \leq j <n \label{def:eps} \\ 
  \omega^k_{m\times n}(\mathcal{B})_{i,j} &=& \mathcal{B}_{[i+k]_l,j}, \quad  0 \leq i <m, 0 \leq j <n \label{def:omega}
\end{eqnarray}
\noindent 
where $[x]_y$ denotes $x \bmod y$. 

The two transformation operators, $\sigma$ and $\tau$, are similar to those introduced in~\cite{jiang2018secure}, but \emph{more general and applicable to an arbitrary shape matrix} instead of a square matrix alone. Essentially, a $\sigma$ transformation rotates each row of a matrix horizontally by its corresponding row index (for example, each element in the $2_{nd}$ row is cyclically rotated 2 positions to the left), and a $\tau$ transformation rotates a column by its corresponding column index. Figures~\ref{Fig.op.1} and \ref{Fig.op.2} illustrate examples of $\sigma$ and $\tau$ transformations. 

We define two \emph{new} transformation operators, $\epsilon^k_{m\times n}$ and $\omega^k_{m\times n}$, with respect to a matrix of arbitrary shape. Given $\mathcal{C}_{m\times n} = \mathcal{A}_{m\times l} \times \mathcal{B}_{l\times n}$, operator $\epsilon^k_{m\times n}(\mathcal{A})$
generates a matrix with size of $m \times n$ from $\mathcal{A}_{m\times l}$ (duplicating or cropping columns when necessary), by shifting matrix $\mathcal{A}_{m\times l}$ to the left for $k$ columns. Similarly, $\omega^k_{m\times n}(\mathcal{B})$ generates a matrix with size of $m \times n$ from $\mathcal{B}_{l\times n}$ (duplicating or cropping rows when necessary), by shifting matrix $\mathcal{B}_{l\times n}$ upward for $k$ rows. Figures~\ref{Fig.op.3} and \ref{Fig.op.4} illustrate the transformation operators  $\epsilon^0_{3\times 5}(\mathcal{A})$, $\epsilon^1_{3\times 5}(\mathcal{A})$, $\omega^0_{3\times 5}(\mathcal{B})$, and  $\omega^1_{3\times 5}(\mathcal{B})$, respectively.

With the operators defined above, we can perform a general MM using the element-wise operations as follows: 

\begin{equation} \label{eqn:elementwise}
\mathcal{A}_{m\times l} \times \mathcal{B}_{l\times n} = \sum_{k=0}^{l-1}(\epsilon^k_{m\times n} \circ \sigma(\mathcal{A}))\odot (\omega^k_{m\times n} \circ \tau(\mathcal{B})),
\end{equation}
where $\circ$ represents the composition operation. Note that the multiplication ({\it i.e.,} $\odot$) in Equation (\ref{eqn:elementwise}) is element-wise and applied to the entire operands. Figure~\ref{Fig:mm-example-1} shows an example of how an MM can be conducted based on Equation~(\ref{eqn:elementwise}). Given two source matrices, i.e., $\mathcal{A}_{5\times 3}\times \mathcal{B}_{3\times 4}$, with $m=5, l=3$, and $n=4$, $\sigma$ and $\tau$ transformations are conducted on $\mathcal{A}$ and $\mathcal{B}$, respectively. Then three iterations of $\epsilon$ and $\omega$ transformations are performed to obtain three partial products, which are accumulated to get the final product. We have the following proof sketch to show that the above method indeed produces the correct MM results for arbitrary matrices. 

\begin{equation} \label{form:proof}
\begin{split}
    &\sum_{k=0}^{l-1}(\epsilon^k_{m\times n} \circ \sigma(\mathcal{A}))_{i,j}\cdot (\omega^k_{m\times n} \circ \tau(\mathcal{B}))_{i,j} \\
    = &\sum_{k=0}^{l-1}\sigma(\mathcal{A})_{i,[j+k]_l}\cdot \tau(\mathcal{B})_{[i+k]_l,j} \\  
    = &\sum_{k=0}^{l-1}\mathcal{A}_{i,[i+j+k]_l}\cdot \mathcal{B}_{[i+j+k]_l,j}  \\ 
    = &\sum_{k=0}^{l-1}\mathcal{A}_{i,k}\cdot \mathcal{B}_{k,j}  = (\mathcal{A}\cdot \mathcal{B})_{i,j}
\end{split}
\end{equation}

\begin{figure}[h]
{
\centering
\subfigure[$\sigma$ operator: rotating $i_{th}$ row left by $i$ slots.]{
\label{Fig.op.1}
\includegraphics[width=0.38\textwidth]{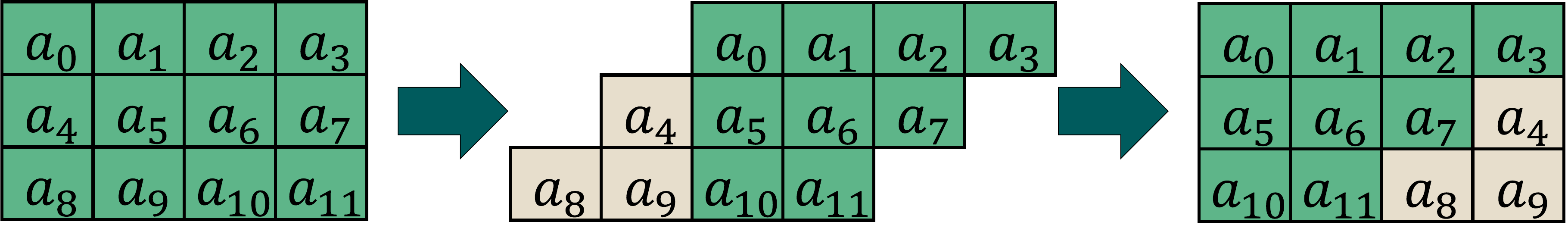}}

\subfigure[$\tau$ operator: rotating $j_{th}$ column upward by $j$ slots.]{
\label{Fig.op.2}
\includegraphics[width=0.28\textwidth]{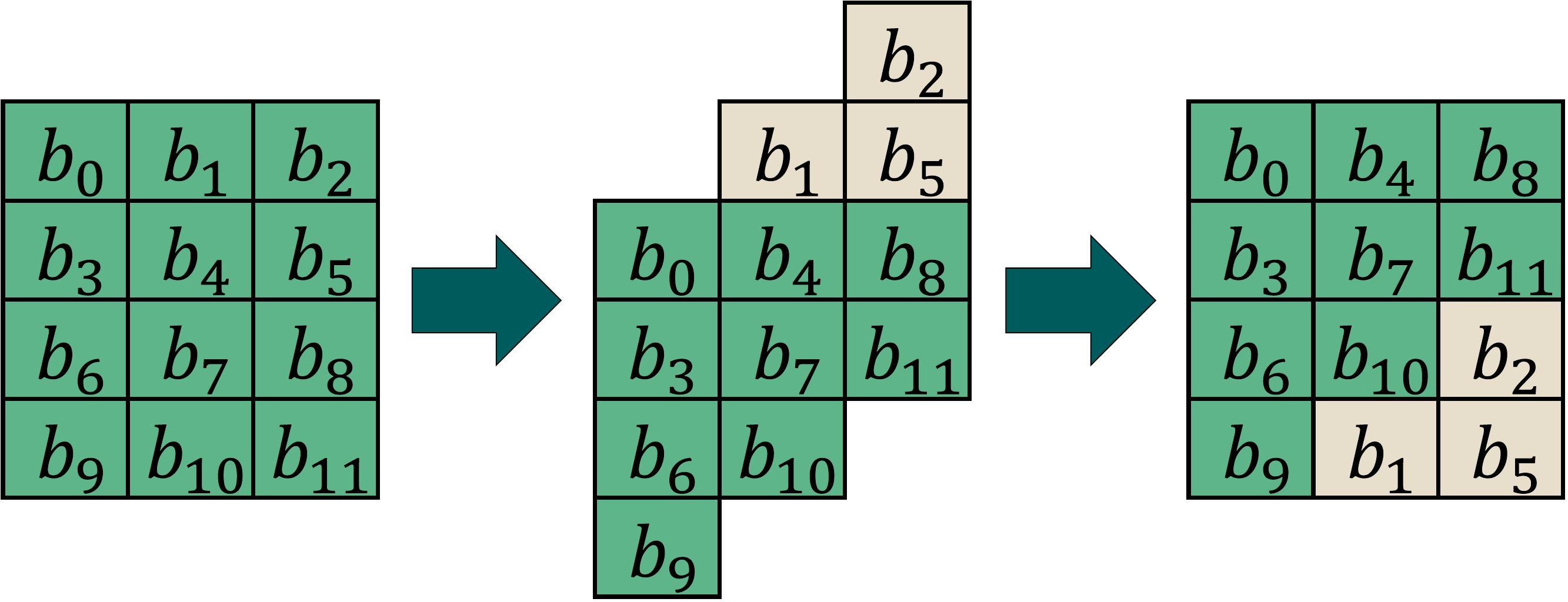}}

\subfigure[$\epsilon^0_{3\times 5}(\mathcal{A}_{3\times2})$ and $\epsilon^1_{3\times 5}(\mathcal{A}_{3\times2})$ with  $\mathcal{C}_{3\times 5} = \mathcal{A}_{3\times 2} \times \mathcal{B}_{2\times5}$.]{
\label{Fig.op.3}
\includegraphics[width=0.42\textwidth]{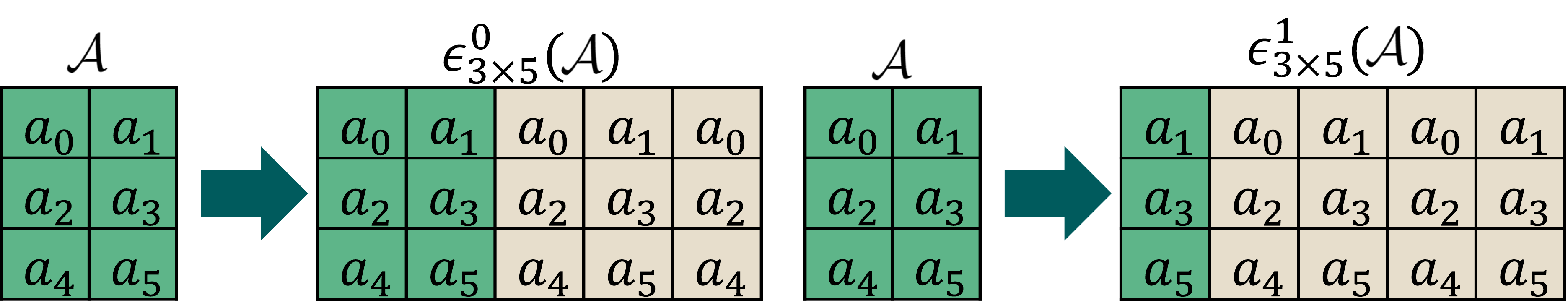}}

\subfigure[$\omega^0_{3\times 5}(\mathcal{B}_{2\times5})$  and $\omega^1_{3\times 5}(\mathcal{B}_{2\times5})$ with  $\mathcal{C}_{3\times 5} = \mathcal{A}_{3\times 2} \times \mathcal{B}_{2\times5}$.]{
\label{Fig.op.4}
\includegraphics[width=0.47\textwidth]{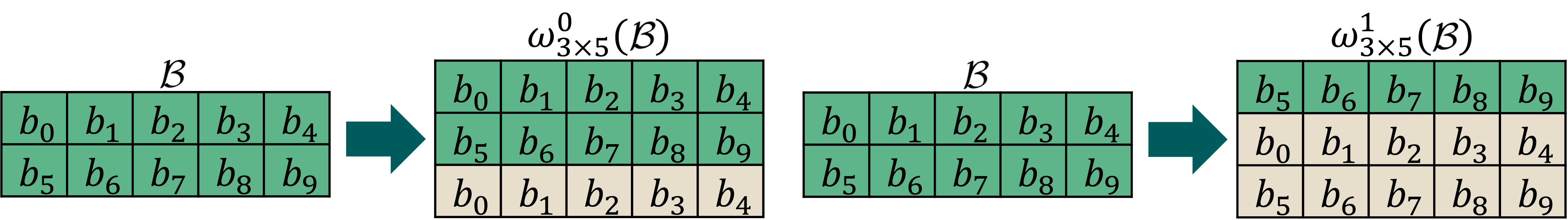}}

\caption{The illustration of $\sigma$, $\tau$, $\epsilon$, and $\omega$ transformation operators}
\label{Fig.op}
}

\setlength{\belowdisplayskip}{10pt}
\end{figure}

\begin{figure}[h]
\centering
\includegraphics[width=0.35\textwidth]{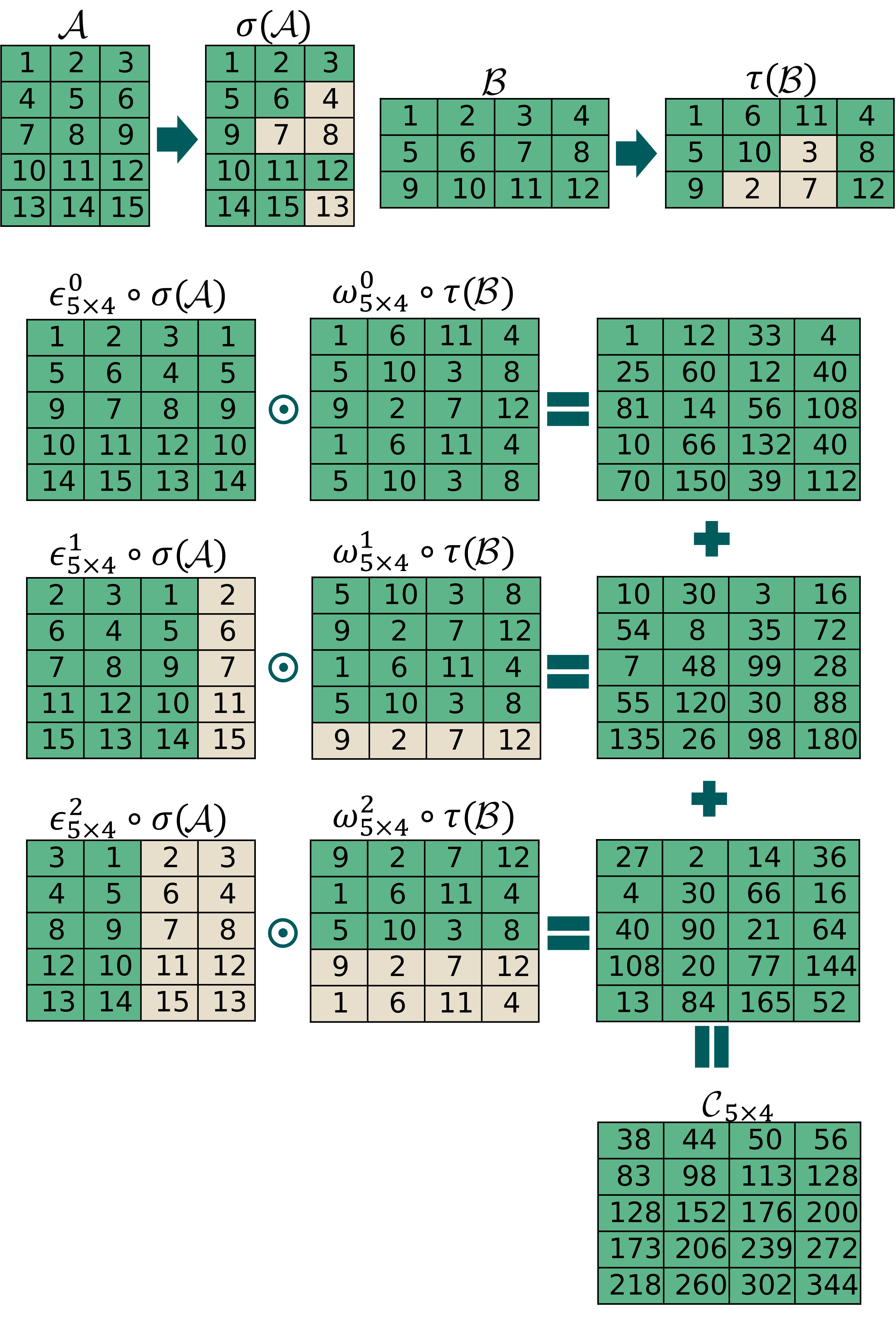}
\caption{An illustration example of the element-wise MM for $\mathcal{A}_{5\times 3}\times \mathcal{B}_{3\times 4}$ with $m=5, l=3$, and $n=4$, $\sigma$ and $\tau$ transformations are first conducted on $\mathcal{A}$ and $\mathcal{B}$, respectively. Then three iterations of $\epsilon$ and $\omega$ transformations are performed to obtain three partial products, which are accumulated to get the final product.}

\label{Fig:mm-example-1}
\end{figure}

\iftrue

\subsection{The \underline{HE}-based \underline{G}eneral \underline{M}atrix \underline{M}ultiplication (HEGMM)}
\label{subsec:HEGMM}

With the element-wise matrix multiplication method introduced above, we are now ready to present our approach for HE matrix multiplication in the SIMD manner. As mentioned before, 

it is critical to minimize the number of HE operations (such as those, especially the HE-Mult, as shown in Table~\ref{tab:cost} ) and thus reduce the computational cost. 
In this subsection, we first introduce how transformations, such as $\sigma, \tau, \epsilon^k_{m\times n}$, and $\omega^k_{m\times n}$, are performed using the primitive HE operations. We then present our first algorithm for HE MM based on the element-wise matrix multiplication strategy presented above.

\subsubsection{Linear Transformation}
\label{subsec:Ltranform}

To perform MM under the HE framework, two-dimensional matrices need to be flattened (either with column-major or row-major order) into one-dimensional ciphertexts, and all operations are performed on the ciphertexts. Therefore, a critical challenge to implement the computational strategy in Equation~(\ref{eqn:elementwise}) is how to efficiently conduct $\sigma, \tau, \epsilon^k_{m\times n}$, and $\omega^k_{m\times n}$, $k=\{0, ..., (l-1)\}$  transformation operations. 
Note that an arbitrary linear transformation over a vector \textit{\textbf{m}}, \ie, $L: \mathcal{R}_{x} \rightarrow \mathcal{R}_{y}$, can be expressed as 
$L: \textit{\textbf{m}} \rightarrow \textit{U}\cdot \textit{\textbf{m}}$, where $\textbf{U} \in \mathcal{R}_{y\times x}$ is the transformation matrix. 
As shown by Halevi and Shoup \cite{halevi2014algorithms}, matrix-vector multiplications can be calculated using the combination of rotation and element-wise multiplication operations. Specifically, for $0 \le z < 
x$, let the $z\text{-}th$ \textit{diagonal} \textit{vector} of \textbf{U} be 

\begin{equation}
\nonumber
   \textit{\textbf{u}}_z=
   \begin{cases}
   \underset{\left | x \right | }{(\underbrace{U_{0,z},U_{1,z+1},...,U_{x-z-1,x-1},0,...,0})} & \text{ if } z \ge 0 \\
   \underset{\left | x \right | }{(\underbrace{0,...,0,U_{z,0},U_{z+1,1},...,U_{y-1,y-z-1}})} & \text{ if } z<0
   \end{cases}
\end{equation}

where $x$ and $y$ are the matrix dimensions and $z$ is the index of diagonal vector.

Then we have

\begin{equation} \label{eqn:lintrans}
   \textbf{U}\cdot \textit{\textbf{m}}=
   \sum_{-y \le z < x}(\textit{\textbf{u}}_z\odot HE\text{-}Rot(\textit{\textbf{m}};z)),
\end{equation}
where $\odot$ denotes the component-wise multiplication.

According to Equation (\ref{eqn:lintrans}), we can construct the transformations defined in Equations (\ref{def:sigma})-(\ref{def:omega}) with flattened matrix $\Tilde{\mathcal{A}}$ and $\Tilde{\mathcal{B}}$ such that
{
\begin{eqnarray}
  \sigma(\Tilde{\mathcal{A}}) &=& \textbf{U}^\sigma \cdot \Tilde{\mathcal{A}} , \label{eqn:sigma}\\
  \tau(\Tilde{\mathcal{B}}) &=& \textbf{U}^\tau \cdot \Tilde{\mathcal{B}}, \label{eqn:tau} \\
  \epsilon^k_{m\times n}(\Tilde{\mathcal{A}}) &=& \textbf{U}^{\epsilon^k_{m\times n}} \cdot \Tilde{\mathcal{A}}, \label{eqn:epsilon} \\
  \omega^k_{m\times n}(\Tilde{\mathcal{B}}) &=& \textbf{U}^{\omega^k_{m\times n}} \cdot \Tilde{\mathcal{B}} \label{eqn:omega}.
\end{eqnarray}
}
Let $\Tilde{\mathcal{A}}$ and $\Tilde{\mathcal{B}}$ be source matrices flattened in the \emph{column-major} order. By generalizing the location change patterns for the operators, we can define $\textbf{U}^\sigma$, $\textbf{U}^\tau$, $\textbf{U}^{\epsilon^k_{m\times n}}$ and $\textbf{U}^{\omega^k_{m\times n}}$ as follows:
\begin{eqnarray} 
    \textbf{U}^\sigma_{i+j\cdot m,h} &=& 
        \begin{cases}
         1 & \text{ if } h= i+ [i+j]_l\cdot m, \\
          0 & \text{ otherwise; } 
        \end{cases} \label{eqn:Usigma}\\
    \textbf{U}^\tau_{i+j\cdot l, h} &=&
        \begin{cases}
          1& \text{ if } h = [i+j]_l + j\cdot l, \\
          0& \text{ otherwise; }
        \end{cases} \label{eqn:Utau}\\
    \textbf{U}^{\epsilon^k_{m\times n}}_{i,j} &=& 
        \begin{cases}
         1 & \text{ if } j=[k\cdot m+i]_{m\cdot l}\\
         0 & \text{ otherwise; } 
        \end{cases} \label{eqn:Uepsc}\\
    \textbf{U}^{\omega^k_{m\times n}}_{i,j} &=&
        \begin{cases}
          1& \text{ if } j=[k+[i]_m]_{l}+\left \lfloor i/m \right \rfloor \cdot l\\
          0& \text{ otherwise; }
        \end{cases} \label{eqn:Uomgc}
\end{eqnarray}
\noindent For the sake of clarity, scopes of $i$, $j$ and $h$ in Equation (\ref{eqn:Usigma})-(\ref{eqn:Uomgc}) are listed below.

\begin{table}[htbp]
  \centering
    \begin{tabular}{c|ccc}
          & \textbf{$i$} & \textbf{$j$} & \textbf{$h$} \\
    \hline
    {$\textbf{U}^\sigma$} & $[0,m)$ & $[0,l)$ & $[0,ml)$ \\
    {$\textbf{U}^\tau$} & $[0,l)$ & $[0,n)$ & $[0,nl)$ \\
    {$\textbf{U}^{\epsilon^k_{m\times n}}$} & $[0,mn)$ & $[0,ml)$ & N/A\\
    {$\textbf{U}^{\omega^k_{m\times n}}$} & $[0,mn)$ & $[0,nl)$ & N/A\\
    \end{tabular}%
  \label{tab:ijrange}%
\end{table}%

When $\Tilde{\mathcal{A}}$ and $\Tilde{\mathcal{B}}$ are matrices flattened in the \emph{row-major} order, similar transformation matrices can be constructed. We  omit it due to page limit.

Note that, from equation (\ref{eqn:lintrans}), the $\sigma, \tau, \epsilon^k_{m\times n}$, and $\omega^k_{m\times n}$ operations can be realized using a sequence of HE-Rot, HE-CMult, and HE-Add operations.
Figure~\ref{Fig.plt-ex} shows examples of transformations $\epsilon^1_{5\times 3}(\mathcal{A})$ and $\omega^1_{3\times 5}(\mathcal{B})$ as well as the associated matrices $\textbf{U}^{\epsilon^1_{5\times 3}}$ and $\textbf{U}^{\omega^1_{3\times 5}}$ for matrix $\mathcal{A}_{5\times 3}$ and $\mathcal{B}_{3\times 5}$, respectively.
In the meantime, equation (\ref{eqn:lintrans}) clearly shows that the computational cost depends heavily on how many diagonal vectors (i.e., $\textit{\textbf{u}}_z$ in equation (\ref{eqn:lintrans})) in the corresponding transformation matrices, i.e., $\textbf{U}^\sigma$, $\textbf{U}^\tau$, $\textbf{U}^{\epsilon^k_{m\times n}}$,  and $\textbf{U}^{\omega^k_{m\times n}}$, are non-zeros. The more the non-zero diagonal vectors are, the higher the computation costs become. To this end, we have the following theorems that reveal important properties related to non-zero diagonal vectors in these transformation matrices.

\begin{theorem}
\label{thrm:sigma-diag}
Let $\sigma(\mathcal{A}) = \textbf{U}^\sigma \mathcal{A}$ for $\mathcal{A}$ with a dimension of $m\times l$. There are at most $2\cdot \min(m,l)-1$ non-zero diagonals in $\textbf{U}^\sigma$ no matter whether the matrix is flattened with a column-major or row-major order. 
\end{theorem}

\begin{theorem}
\label{thrm:tau-diag}
Let $\tau(\mathcal{B}) = \textbf{U}^\tau \mathcal{B}$ for $\mathcal{B}$ with a dimension of $l\times n$. There are at most $2\cdot \min(n,l)-1$ non-zero diagonals in $\textbf{U}^\tau$ no matter if the matrix is flattened with a column-major or row-major order. 
\end{theorem}

\begin{theorem}
\label{thrm:eps-diag}
Let $\epsilon^k_{m\times n}(\mathcal{A}) =\textbf{U}^{\epsilon^k_{m\times n}} \mathcal{A}$ be the linear transformation $\epsilon_{m\times n}: \mathcal{R}_{m\times l} \rightarrow \mathcal{R}_{m\times n}$ with matrix $\mathcal{A}$ having a dimension of $m\times l$. 
There are at most $\left \lfloor \frac{n}{l} \right \rfloor +1$ non-zero diagonal vectors in $\textbf{U}^{\epsilon^k_{m\times n}}$ when the matrix is flattened with the \textbf{column-major} order; 
There are at most $(\left \lfloor \frac{n}{l} \right \rfloor +2)\cdot m$ non-zero diagonal vectors in $\textbf{U}^{\epsilon^k_{m\times n}}$ when matrix $\mathcal{A}$ is flattened with the \textbf{row-major} order. 
Specifically, when $n=l$, there are no more than 2 non-zero diagonals in $\textbf{U}^{\epsilon^k_{m\times n}}$, no matter if the matrix is flattened in column-major or row-major order.
\end{theorem}

\begin{theorem}
\label{thrm:omg-diag}
Let $\omega^k_{m\times n}(\mathcal{B}) = \textbf{U}^{\omega^k_{m\times n}} \mathcal{B}$ be the linear transformation $\omega_{m\times n}: \mathcal{R}_{l\times n} \rightarrow \mathcal{R}_{m\times n}$ with matrix $\mathcal{B}$ having a dimension of $l\times n$. 
There are at most $(\left \lfloor \frac{m}{l} \right \rfloor+2)\cdot n$ non-zero diagonal vectors in $\textbf{U}^{\omega^k_{m\times n}}$ when the matrix is flattened with \textbf{column-major} order; There are
at most $\left \lfloor \frac{m}{l} \right \rfloor +1$ non-zero diagonal vectors in $\textbf{U}^{\omega^k_{m\times n}}$ when matrix $\mathcal{B}$ is flattened with \textbf{row-major} order. 
Specifically, when $m=l$, there are no more than 2 non-zero diagonals in $\textbf{U}^{\omega^k_{m\times n}}$, no matter if the matrix is flattened in column-major or row-major order.
 \end{theorem}

\begin{figure}[h]
{
\centering
\subfigure[The process of linear transformation for $\epsilon^1_{5\times 3}$ transformation on matrix $\mathcal{A}_{5\times 3}$. $\textit{\textbf{u}}_5$ is the vector with 10 1's and 5 0's while $\textit{\textbf{u}}_{\text{-}10}$ with 10 0's and 5 1's. Therefore, according to Equation~\ref{eqn:lintrans}, it rotates flattened $\Tilde{\mathcal{A}}$ 5 slots and times $\textit{\textbf{u}}_5$. Then it rotates flattened $\Tilde{\mathcal{A}}$ 10 slots reversely and times $\textit{\textbf{u}}_{\text{-}10}$. Finally, add all above partial products together.]{
\includegraphics[width=0.47\textwidth]{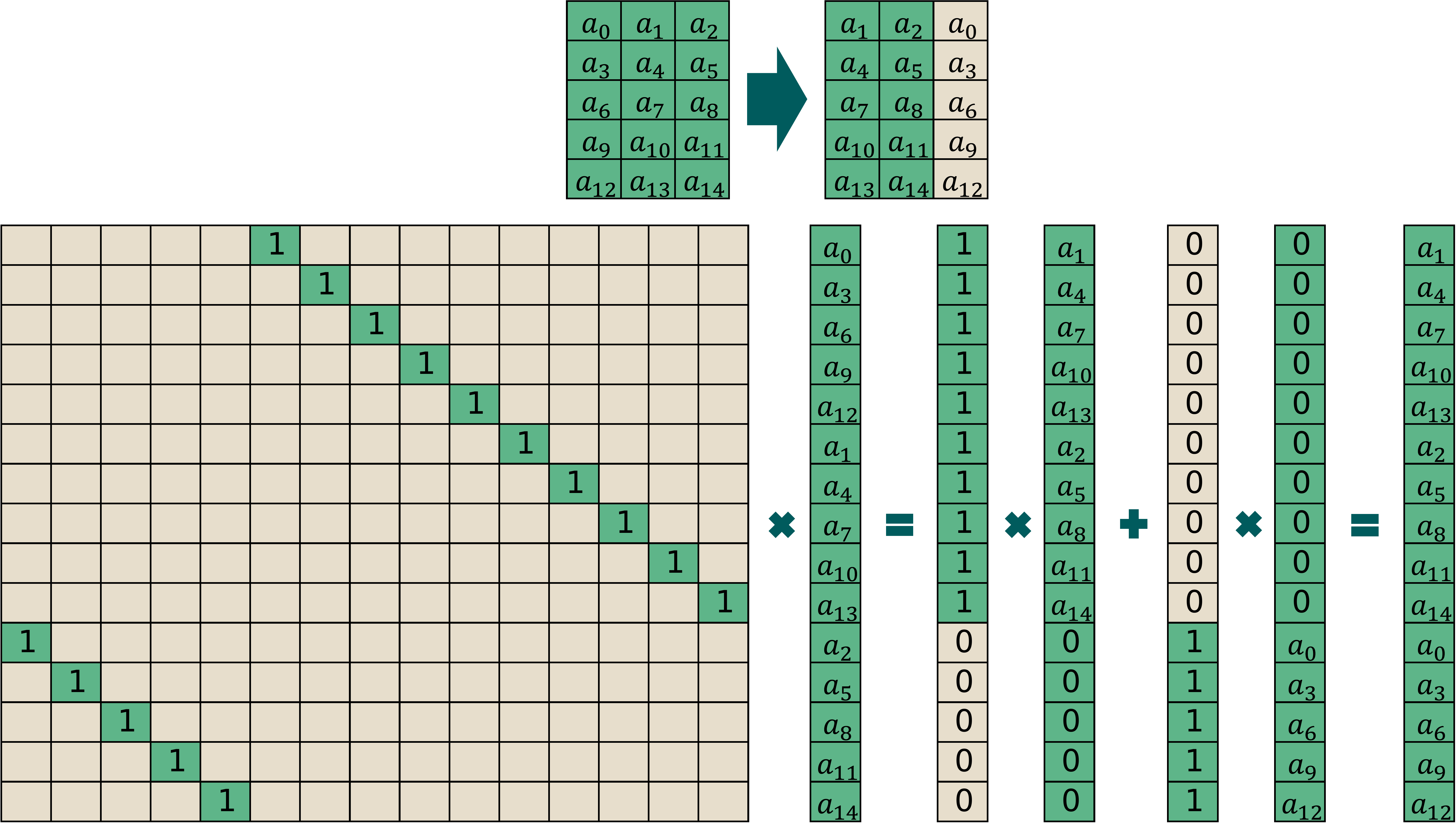}}

\centering
\subfigure[The process of linear transformation for $\omega^1_{3\times 5}$ transformation on matrix $\mathcal{B}_{3\times 5}$. $\textit{\textbf{u}}_1$ is the vector with ten 1's and five 0's while $\textit{\textbf{u}}_{\text{-}2}$ with 10 0's and 5 1's. Therefore, according to Equation~\ref{eqn:lintrans}, it rotates flattened $\Tilde{\mathcal{B}}$ 1 slots and times $\textit{\textbf{u}}_1$. Then it rotates flattened $\Tilde{\mathcal{A}}$ 2 slots reversely and times $\textit{\textbf{u}}_{\text{-}2}$. Finally, add all above partial products together.]{
\includegraphics[width=0.47\textwidth]{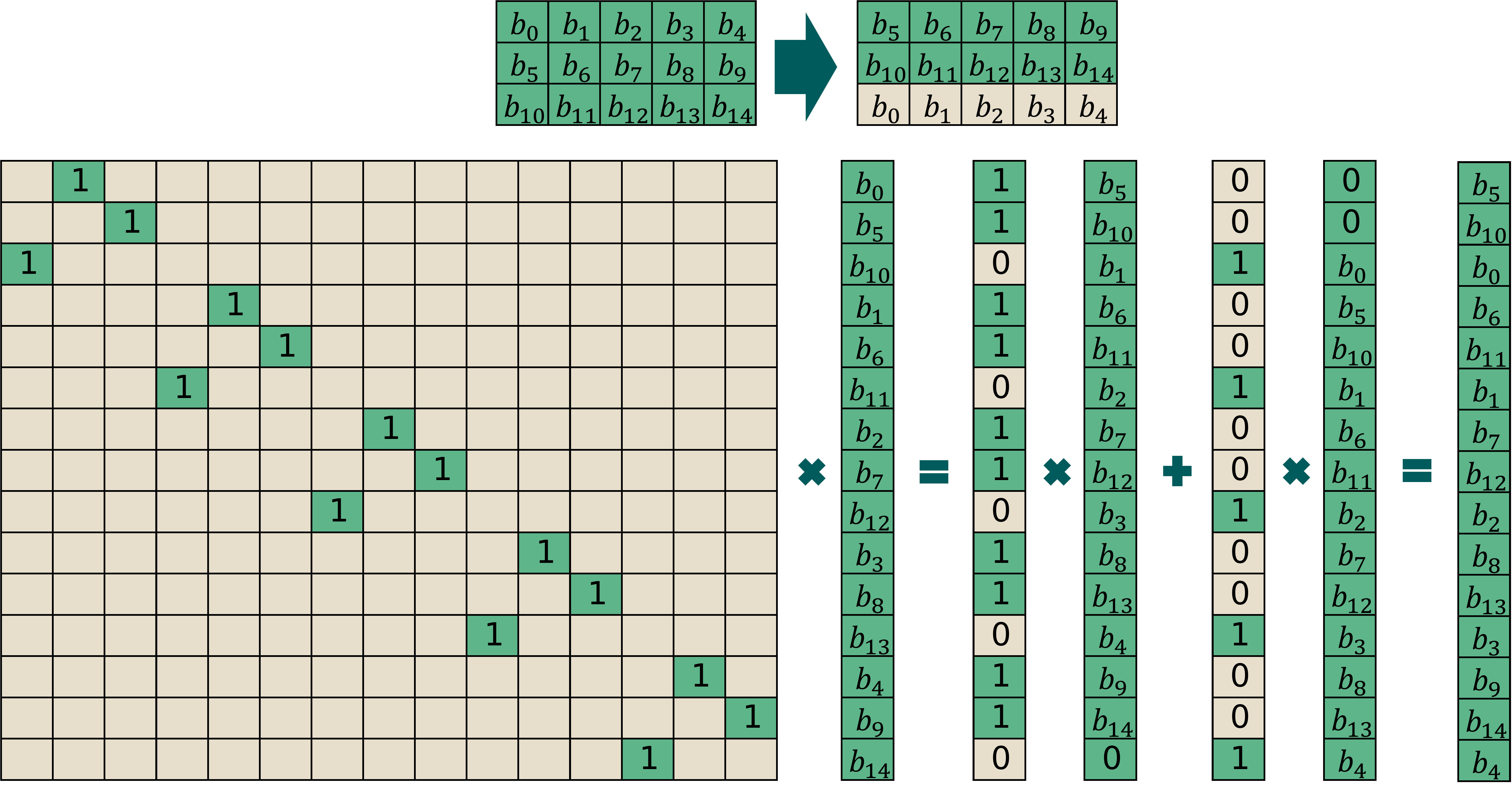}}
}
\caption{The permutation matrices $\textbf{U}^{\epsilon^1_{5\times 3}}$ and $\textbf{U}^{\omega^1_{3\times 5}}$ and linear transformations of $\epsilon_{5\times 3}^1(\mathcal{A}_{5\times 3})$ and $\omega_{3\times 5}^1(\mathcal{B}_{5\times 3})$.}
\label{Fig.plt-ex}

\setlength{\belowdisplayskip}{10pt}
\end{figure}

The proofs for Theorems~\ref{thrm:sigma-diag}-\ref{thrm:omg-diag} can be found in the appendix~\ref{apx:proof}. 

According to Theorems~\ref{thrm:sigma-diag} and ~\ref{thrm:tau-diag}, the numbers of non-zero diagonal vectors in $\textbf{U}^\sigma$ and $\textbf{U}^\tau$ depend solely on the dimensions of corresponding matrices and are independent of how the matrices are flattened. However, as shown in Theorems~\ref{thrm:eps-diag} and ~\ref{thrm:omg-diag}, the numbers of non-zero diagonal vectors in $\textbf{U}^{\epsilon^k_{m\times n}}$ and $\textbf{U}^{\omega^k_{m\times n}}$ depend on not only the dimensions of matrices but also the way they are flattened. 
\begin{center}
\begin{algorithm}[hbt!]
\caption{HEGMM: HE-based General Matrix Multiplication}
\label{alg:HEGMM}
\LinesNumbered
\KwIn{matrix $\mathcal{A}_{m\times l}$ and matrix $\mathcal{B}_{l\times n}$}
\KwOut{$\mathcal{C}_{m\times n}$}
\textbf{[Step1]}

$ct.{\mathcal{{A}}}^{(0)}  \gets \sigma({\mathcal{{A}}})$

$ct.{\mathcal{B}}^{(0)}  \gets \tau({\mathcal{B}})$

\textbf{[Step2]}

\For{$k=0$ to $l-1$}
{
    $ct.\mathcal{{A}}^{(k)}  \gets {\epsilon^k_{m\times n}}(ct.\mathcal{{A}}^{(0)})$
    
    $ct.\mathcal{B}^{(k)}  \gets {\omega^k_{m\times n}}(ct.\mathcal{B}^{(0)})$
    
    $ct.\mathcal{C}  \gets ct.\mathcal{C}+ct.\mathcal{{A}}^{(k)} \odot ct.\mathcal{B}^{(k)}$
}

\textbf{[Step3]}

$\mathcal{C}_{m\times n} \gets ct.\mathcal{C}$

\Return{ $\mathcal{C}_{m\times n}$ }
\end{algorithm} 
\end{center}

\subsubsection{The HEGMM Algorithm}
HEGMM is a straightforward implementation of Equation (\ref{eqn:elementwise}). We first conduct $\sigma$ and $\tau$ transformations (lines 2-3) on the source matrices of $\mathcal{A}$ and $\mathcal{B}$. 
We then go through a loop (lines 5-9) that apply $\epsilon^k_{m\times n}$ and $\tau^k_{m\times n}$ transformations and element-wise multiplication and addition to calculate and accumulate the partial product. The final result can be obtained by decrypting the sum of the product (line 11).  
 
The computational complexity of Algorithm~\ref{alg:HEGMM} mainly comes from the required HE operations associated with the $\sigma, \tau, \epsilon^k_{m\times n}$, and $\omega^k_{m\times n}$ operations. Assuming $\mathcal{A}$ and $\mathcal{B}$ are encrypted, from Theorem~\ref{thrm:sigma-diag} and Theorem~\ref{thrm:tau-diag}, we know that there are $2\min(m,l)-1$ (resp. $2\min(n,l)-1$) non-diagonals for the $\sigma$ (resp. $\tau$) operation. Therefore, according to equation~(\ref{eqn:lintrans}), the $\sigma$ (resp. $\tau$) operation requires $2\min(m,l)-1$ (resp. $2\min(n,l)-1$) HE-CMult, HR-Rot, and HR-Add operations. These computational costs have nothing to do with how the matrices are flattened (e.g., in column-major or in row-major order), and they also become trivial if they are performed on $\mathcal{A}$ and $\mathcal{B}$ in plaintext. However, the $\epsilon$ and $\omega$ operations, which must be performed multiple times in the cloud, require HE operations depending on not only the dimensions of matrices but also the way they are flattened. As a result, the computational complexities can be dramatically different under different scenarios, as shown in Theorems~\ref{thrm:eps-diag} and ~\ref{thrm:omg-diag}. In the next sub-section, we show how we can take advantage of this fact to reduce the computational cost effectively. 

\subsection{The Enhanced HEGMM Algorithm}
\label{subsec:HEGMM-Enhanced}

In this section, we introduce a more elaborated approach for HEGMM that can be more computationally efficient. The fundamental principle we rely on to develop this algorithm is presented in Theorem~\ref{thrm:eps-diag} and ~\ref{thrm:omg-diag}. For the HE matrix multiplication of $\mathcal{A}_{m\times l} \times \mathcal{B}_{l\times n}$, the proposed new algorithm can significantly improve the computation efficiency when $m=\min\{m,l,n\}$ and/or $n=\min\{m,l,n\}$. If not, we can always resort to Algorithm~\ref{alg:HEGMM} to find the solution. Therefore, in what follows, we first discuss the new approach based on two cases: (\emph{i}) $m=\min\{m,l,n\}$; and (\emph{ii}) $n=\min\{m,l,n\}$. We then present the algorithm and related discussions in detail.

\subsubsection{$m=\min\{m,l,n\}$}

\begin{figure}[h]
\centering
\includegraphics[width=0.49\textwidth]{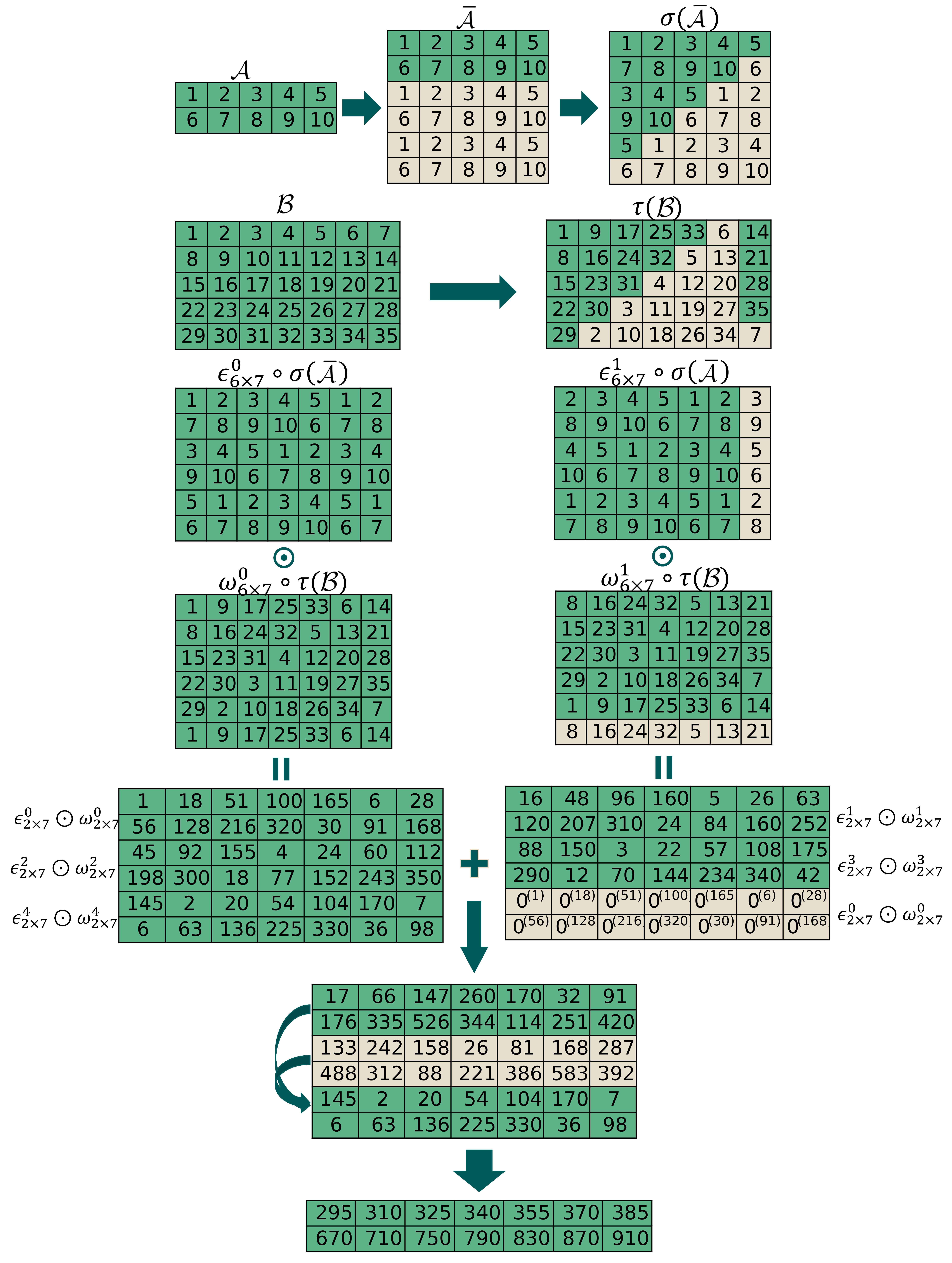}
\caption{An illustration example of the Enhanced HEGMM Algorithm for multiplying two matrices $\mathcal{A}_{2\times 5}$ and $\mathcal{B}_{5\times 7}$, with $m=2$, $l=5$ and $n=7$. $\Bar{\mathcal{A}}$ is the matrix by duplicating $\mathcal{A}$ 3 times, i.e., $t=\left \lceil 5/2 \right \rceil=3$ and $\mathcal{B}_{5\times 7}$ remains unchanged. The partial products are accumulated to obtain the final product. Note $\epsilon^0_{2\times 7}(\sigma(\mathcal{A}))\odot \omega^0_{2\times 7}(\tau(\mathcal{B}))$ is generated twice, and the duplicated partial products should be excluded from the final results. }

\label{Fig.example-dupA}
\end{figure}

For the HE MM of $\mathcal{A}_{m\times l} \times \mathcal{B}_{l\times n}$, Algorithm~\ref{alg:HEGMM} needs to perform $l$ iterations, with each iteration including one $\epsilon$ transformation, one $\omega$ transformation, one HE-Add, and one HE-Mult operation. Assuming the matrix is flattened with the \textbf{\emph{column-major}} order, according to Theorem~\ref{thrm:eps-diag} and ~\ref{thrm:omg-diag}, one $\epsilon$ transformation and one $\omega$ transformation would result in no more than $((\left \lfloor \frac{n}{l} \right \rfloor +1) + 2n)$ non-zero diagonals in corresponding transformation matrices, with each non-zero diagonal requiring one HE-Add, one HE-Rot, and one HE-CMult operations. However, if we can expand matrix $\mathcal{A}_{m\times l}$ to $\bar{\mathcal{A}}_{l\times l}$, then the number of non-zero diagonals becomes no more than $((\left \lfloor \frac{n}{l} \right \rfloor +1) + 2)$ instead. Since $n \geq 1$ and
\[ ((\left \lfloor \frac{n}{l} \right \rfloor +1) + 2n) \geq ((\left \lfloor \frac{n}{l} \right \rfloor +1) + 2),\]
the number of non-zero diagonals and, thus, the computational cost can be dramatically reduced. 

Note that, if we assume the matrix is flattened with the \textbf{\emph{row-major}} order, one $\epsilon$ transformation and one $\omega$ transformation would result in no more than $((\left \lfloor \frac{n}{l} \right \rfloor +2)m + 1)$ non-zero diagonals in corresponding transformation matrices. When expanding matrix $\mathcal{A}_{m\times l}$ to $\bar{\mathcal{A}}_{l\times l}$, the total number of non-zero diagonals in the corresponding transformation matrices becomes $((\left \lfloor \frac{n}{l} \right \rfloor +2)l + 1)$, according to Theorems~\ref{thrm:eps-diag} and ~\ref{thrm:omg-diag}. It becomes obvious that using the \textbf{\emph{column-major}} order is a better choice than using the \textbf{\emph{row-major}} order in this case, since when $m > 1$, we have
\[((\left \lfloor \frac{n}{l} \right \rfloor +2)m + 1) \geq ((\left \lfloor \frac{n}{l} \right \rfloor +1) + 2).\]

The question now becomes how to expand $\mathcal{A}_{m\times l}$ to $\bar{\mathcal{A}}_{l\times l}$ and maintain the logical correctness of the MM result. One intuitive approach is to expand $\mathcal{A}_{m\times l}$ by filling zeroes to the newly added elements, i.e., the elements in rows from row $m$ to $(l-m-1)$. The final product, as a sub matrix, can be extracted from the product of $\bar{\mathcal{A}}_{l\times l} \times \mathcal{B}_{l\times n}$ easily. Note that expanding the matrix dimension does not increase the computational complexity in the SIMD scheme as long as the result matrix can fit in one message.  

Rather than simply filling zeroes, we can expand $\mathcal{A}_{m\times l}$ by duplicating the rows of $\mathcal{A}_{m\times l}$ repeatedly. This helps to reduce the number of iterations (lines 5-9) in 
Algorithm~\ref{alg:HEGMM}, thanks to our observations that are formulated in the following theorem.  
\begin{theorem} \label{thrm:dupA}
Let $\mathcal{A}_{m\times l}$ and $\mathcal{B}_{l\times n}$ with $m < l$, and let $\bar{A}$ be matrix expanded with $t=\left \lceil \frac{l}{m} \right \rceil$ copies of $\mathcal{A}$ vertically, i.e., $\bar{\mathcal{A}} = \{\bar{A_0}; \bar{A_1}; ...; \bar{A}_{(t-1)}\}^T$ 
with $\bar{A_0}=\bar{A_1}=...= \bar{A}_{(t-1)}=\mathcal{A}_{m\times l}$. Then
\begin{itemize}
\item $\epsilon^k_{tm\times n} ( \sigma(\bar{\mathcal{A}}))\odot \omega^k_{tm\times n} ( \tau(\mathcal{B}))$ contains $t$ items of $\epsilon^p_{m\times n} ( \sigma(\mathcal{A}))\odot \omega^p_{m\times n} ( \tau(\mathcal{B}))$, with $p \in \{[k]_l,[k+m]_l,..., [k+(t-1)m]_l\}$.
\item $\epsilon^k_{tm\times n} ( \sigma(\bar{\mathcal{A}}))\odot \omega^k_{tm\times n} ( \tau(\mathcal{B}))$, $k=0,1,...,(m-1)$ contains all items of $\epsilon^p_{m\times n} (\sigma(\mathcal{A}))\odot \omega^p_{m\times n} (\tau(\mathcal{B}))$, with $p\in \{0, 1, ..., (l-1)\}$.
\end{itemize}
\end{theorem}

According to Theorem~\ref{thrm:dupA}, after expanding $\mathcal{A}_{m\times l}$ with $t$ copies of 
$\mathcal{A}_{m\times l}$ vertically to form $\bar{\mathcal{A}}_{tm\times l}$, each iteration of Algorithm~\ref{alg:HEGMM} can now produce $t$ partial products $\epsilon^p_{m\times n}({\mathcal{A}})\odot \omega^p_{m\times n}({\mathcal{B}})$. As a result, the required HE computations can be greatly reduced, which can be better illustrated using the example in Figure ~\ref{Fig.example-dupA}.  

Figure~\ref{Fig.example-dupA} shows two source matrices $\mathcal{A}_{2\times 5}$ and $\mathcal{B}_{5\times 7}$, with $m=2$, $l=5$ and $n=7$. $\Bar{\mathcal{A}}$ is the matrix by duplicating $\mathcal{A}$ three times, i.e., $t=\left \lceil 5/2 \right \rceil=3$. Note that, each  
$\epsilon_{6\times 7}(\sigma(\bar{\mathcal{A}}))\odot \omega_{6\times 7}(\tau(\mathcal{B}))$ contains
three copies of $\epsilon_{2\times 7}(\sigma(\mathcal{A}))\odot \omega_{2\times 7}(\tau(\mathcal{B}))$, as shown in the figure: $\epsilon^0_{6\times 7}(\sigma(\bar{\mathcal{A}}))\odot \omega^0_{6\times 7}(\tau(\mathcal{B}))$ contains $\epsilon^0_{2\times 7}(\sigma(\mathcal{A}))\odot \omega^0_{2\times 7}(\tau(\mathcal{B}))$, $\epsilon^2_{2\times 7}(\sigma(\mathcal{A}))\odot \omega^2_{2\times 7}(\tau(\mathcal{B}))$, and $\epsilon^4_{2\times 7}(\sigma(\mathcal{A})\odot \omega^4_{2\times 7}(\tau(\mathcal{B}))$. We then need to add all the partial products together to get the final result.

As such, by duplicating $\mathcal{A}_{m\times l}$ into $\bar{\mathcal{A}}_{tm\times l}$, we can reduce not only the HE operations associated with the $\epsilon$ and $\omega$ operations but also the HE-Mult operations (i.e., at most $m$ HE-Mult operations according to Theorem~\ref{thrm:dupA}) for partial production calculations, which is highly costly. Even though extra HE rotations are needed to extract the partial results, the computation cost is much smaller than that of HE-Mult as shown in Table~\ref{tab:cost}. It is worth mentioning that, while one $\epsilon^k_{m\times n}(\bar{\mathcal{A}})\odot \omega^k_{m\times n}({\mathcal{B}})$ helps to produce multiple copies of 
$\epsilon^p_{m\times n}({\mathcal{A}})\odot \omega^p_{m\times n}({\mathcal{B}})$, as shown in Figure~\ref{Fig.example-dupA}, some of them may be produced repeatedly. These redundant copies should be identified, which can be easily identified according to Theorem~\ref{thrm:dupA}, and excluded from the final results. 

\subsubsection{$n=\min\{m,l,n\}$}
\label{subsec:HEGMM-dupB}

We can employ the same analysis flow as above. There are two major differences compared with the case of $m=\min\{m,l,n\}$. (\emph{i}) We duplicate matrix $\mathcal{B}$ \emph{horizontally} to expand $\mathcal{B}$ instead of $\mathcal{A}$; (\emph{ii}) The \textbf{\emph{row-major}} order is a better option than the \textbf{\emph{column-major}} order in this case. 

When $n=\min\{m,l,n\}$, if the matrix is flattened with the \textbf{\emph{row-major}} order, according to Theorem~\ref{thrm:eps-diag} and ~\ref{thrm:omg-diag}, one $\epsilon$ transformation and one $\omega$ transformation would result in no more than $(2m+\left \lfloor \frac{m}{l} \right \rfloor +1)$ non-zero diagonals in corresponding transformation matrices. When expanding $\mathcal{B}_{l\times n}$ to $\bar{\mathcal{B}}_{l\times l}$, the total number of non-zero diagonals in corresponding transformation matrices is reduced to $(2+\left \lfloor \frac{m}{l} \right \rfloor +1)$ instead. However, if the \textbf{\emph{column-major order}} is used, the total number of non-zero diagonals after expanding is $(1 + (2 + \left \lfloor \frac{m}{l} \right \rfloor)\cdot n)$, which makes the \textbf{\emph{row-major}} order a better option to flatten the matrices. 

Similarly, when we expand $\mathcal{B}$ by duplicating $\mathcal{B}_{l\times n}$, we can generate multiple partial products, i.e., $\epsilon^p_{m\times n}({\mathcal{A}})\odot \omega^p_{m\times n}({\mathcal{B}})$, using one HE-Mult operation, as supported by the following theorem. The proof is quite similar to that for Theorem~\ref{thrm:dupA} and thus omitted due to page limit. 

\begin{theorem} \label{thrm:dupB}
Let $\mathcal{A}_{m\times l}$ and $\mathcal{B}_{l\times n}$ with $n < l$, and let $\bar{\mathcal{B}}$ be matrix expanded with $t=\left \lceil \frac{l}{n} \right \rceil$ copies of $\mathcal{B}$ horizontally, i.e., $\bar{\mathcal{B}} = \{\mathcal{B}; \mathcal{B}; ...; \mathcal{B}\}$. Then 

\begin{itemize}
\item {$\epsilon^k_{m\times tn}( \sigma(\mathcal{A}))\odot \omega^k_{m\times tn}( \tau(\bar{\mathcal{B}}))$ contains $t$ items of $\epsilon^p_{m\times n}( \sigma(\mathcal{A}))\odot \omega^p_{m\times n}( \tau(\mathcal{B}))$, with $p=[k]_l, [k+n]_l,..., [k+(t-1)n]_l$;}

\item {$\epsilon^k_{m\times tn}( \sigma(\mathcal{A}))\odot \omega^k_{m\times tn}( \tau(\bar{\mathcal{B}})$, $k=0,1,...,(n-1)$ contains all items of $\epsilon^p_{m\times n}( \sigma(\mathcal{A}))\odot \omega^p_{m\times n}( \tau(\mathcal{B}))$, with $p=0, 1, ..., (l-1)$.}
\end{itemize}
\end{theorem}

\begin{figure}[h]
\centering
\includegraphics[width=0.27\textwidth]{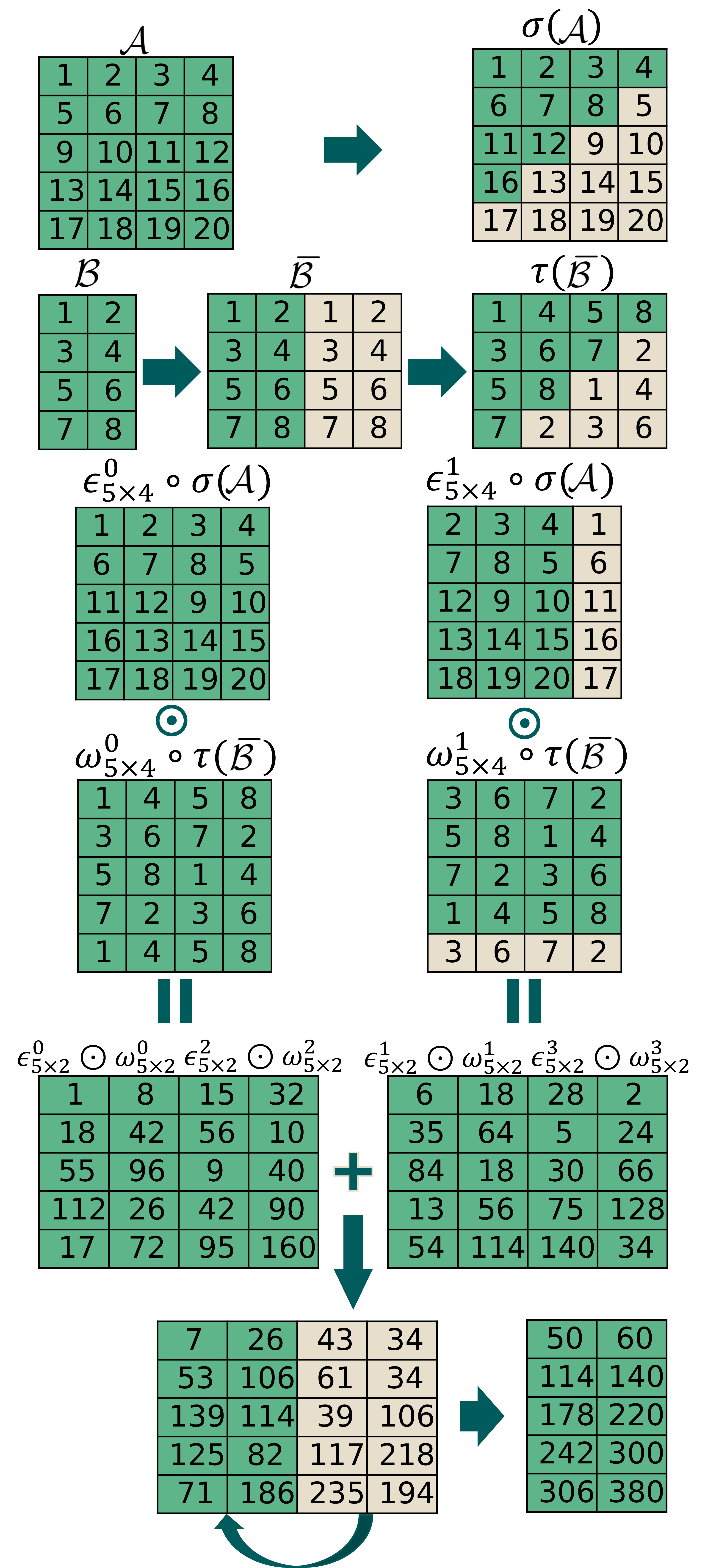}
\caption{This is an illustrative example of the enhanced HE MM algorithm for multiplying two matrices $\mathcal{A}_{5\times 4}$ and $\mathcal{B}_{4\times 2}$, with $m=5$, $l=4$ and $n=2$. $\Bar{\mathcal{B}}$ is the matrix by duplicating $\mathcal{B}$ horizontally for two times, i.e., $t=\left \lceil 4/2 \right \rceil=2$ and $\mathcal{A}_{5\times 4}$ remains unchanged. The partial products are accumulated to obtain the final product.}
\label{Fig.example-dupB}

\setlength{\belowdisplayskip}{10pt}
\end{figure}

Figure~\ref{Fig.example-dupB} shows an illustrative example of HE MM with two source matrices $\mathcal{A}_{5\times 4}$ and $\mathcal{B}_{4\times 2}$, with $m=5$, $l=4$ and $n=2$. $\Bar{\mathcal{B}}$ is the matrix by duplicating $\mathcal{B}$ horizontally for two times, i.e., $t=\left \lceil 4/2 \right \rceil=2$. Note that, each  
$\epsilon_{5\times 4}(\sigma(\mathcal{A}))\odot \omega_{5\times 4}(\tau(\Bar{\mathcal{B}}))$ using one HE-Mult operation can produce 
two copies of $\epsilon_{5\times 2}(\sigma(\mathcal{A}))\odot \omega_{5\times 2}(\tau(\mathcal{B}))$, as shown in the figure: $\epsilon^0_{5\times 4}(\sigma(\mathcal{A}))\odot \omega^0_{5\times 4}(\tau(\bar{\mathcal{B}}))$ contains $\epsilon^0_{5\times 2}(\sigma(\mathcal{A}))\odot \omega^0_{5\times 2}(\tau(\mathcal{B}))$ and 
$\epsilon^2_{5\times 2}(\sigma(\mathcal{A}))\odot \omega^2_{5\times 2}(\tau(\mathcal{B}))$. $\epsilon^1_{5\times 4}(\sigma(\mathcal{A}))\odot \omega^1_{5\times 4}(\sigma(\bar{\mathcal{B}}))$ contains $\epsilon^1_{5\times 2}(\sigma(\mathcal{A}))\odot \omega^1_{5\times 2}(\tau(\mathcal{B}))$ and 
$\epsilon^3_{5\times 2}(\sigma(\mathcal{A}))\odot \omega^3_{5\times 2}(\tau(\mathcal{B}))$. We then need to add all the partial products together to get the final result. As such, we only need to perform at most $n$ HE-Mult operations according to Theorem~\ref{thrm:dupB} to obtain all the partial products. Redundant copies may also be generated during this process, which should be identified according to Theorem~\ref{thrm:dupB} and excluded from the final results. 

\begin{center}
\begin{algorithm}[hbt!]
\caption{HEGMM-Enhanced}
\label{alg:HEGMM-enhanced}
\LinesNumbered
\KwIn{matrix $\mathcal{A}_{m\times l}$, $\mathcal{B}_{l\times n}$}
\KwOut{matrix $\mathcal{C}_{m\times n} =  \mathcal{A}_{m\times l} \times \mathcal{B}_{l\times n}$}

\setcounter{AlgoLine}{0}

$p\gets \min(m,l,n)$

$t\gets \left \lceil l/p \right \rceil$

\tcp{Determine $M$ and $N$ by shape}

\uIf{p=m}
    {
        $\Bar{\mathcal{A}}_{M\times l} \gets \underset{t}{[\underbrace{\mathcal{A};\mathcal{A};\mathcal{A};...;\mathcal{A}}]^T}$
        
        $\Bar{\mathcal{B}}_{l\times N} \gets \mathcal{B}$
        
        $M=t \times m, N=n$
    }
\uElseIf{p=n}
    {
        $\Bar{\mathcal{A}}_{M\times l} \gets \mathcal{A}$
        
        $\Bar{\mathcal{B}}_{l\times N} \gets \underset{t}{[\underbrace{\mathcal{B};\mathcal{B};\mathcal{B};...;\mathcal{B}}]}$

        $M=m, N=t \times n$
    }
\Else
    {
        $\Bar{\mathcal{A}}_{M\times l} \gets \mathcal{A}$
        
        $\Bar{\mathcal{B}}_{l\times N} \gets \mathcal{B}$ 

        $M=m, N=n$
     }

\tcp{Prepossessing on client}
$ct.s\Bar{\mathcal{A}}  \gets \epsilon^0_{M\times \max(l,N)}(\sigma(\Bar{\mathcal{A}}))$ 

$ct.t\Bar{\mathcal{B}}  \gets \omega_{\max(l,M)\times N}^0(\tau(\Bar{\mathcal{B}}))$ 

\tcp{Multiplication on Cloud}
$ct.\mathcal{C}_{m\times n} \gets ct.s\Bar{\mathcal{A}} \odot ct.t\Bar{\mathcal{B}}$

\For{$k=0,1,...,(p-1)$}
    {
        $ct.\Bar{\mathcal{A}}^{(k)} \gets \epsilon_{M\times N}^k(ct.s\Bar{\mathcal{A}})$ 
        
        $ct.\Bar{\mathcal{B}}^{(k)} \gets \omega_{M\times N}^k(ct.t\Bar{\mathcal{B}})$ 

        $ct.\mathcal{C}_{temp} \gets ct.\Bar{\mathcal{A}}^{(k)} \odot ct.\Bar{\mathcal{B}}^{(k)}$ 
        
        \tcp{\small \textit{$ct.\mathcal{C}_{temp}$ contains of $t$ items of $\epsilon^{k+i\cdot p}_{m\times n}(\sigma(\mathcal{A}))\odot \omega^{k+i\cdot p}_{m\times n}(\tau(\mathcal{B})) (0\le i<t)$.}}
        
        \For{$i=0,1,...,(t-1)$}
            {
                $j = [k+i\cdot p]_l$
            
            \uIf{$\epsilon^j_{m\times n}(\sigma(\mathcal{A}))\odot \omega^j_{m\times n}(\tau(\mathcal{B})) \in ct.\mathcal{C}_{temp}$ has not been accumulated in $ct.\mathcal{C}_{m\times n}$}
                {
                $ct.\mathcal{C}_{m\times n} \gets ct.\mathcal{C}_{m\times n}  + \epsilon^j_{m\times n}(\sigma(\mathcal{A}))\odot \omega^j_{m\times n}(\tau(\mathcal{B}))$   
                }
            }
     }

\tcp{Return encrypted result to client}
\Return{$ct.\mathcal{C}_{m\times n}$}
\end{algorithm} 
\end{center}

The overall algorithm for the enhanced HE-based General MM, named \emph{HEGMM-En}, is presented in Algorithm~\ref{alg:HEGMM-enhanced}. Note that, when $m < l$ and $n<l$, we can choose to duplicate either $\mathcal{A}$ or $\mathcal{B}$. In Algorithm~\ref{alg:HEGMM-enhanced}, we choose the smaller of $m$ and $n$ and expand either $\mathcal{A}$ or $\mathcal{B}$ accordingly (lines 3 - 10). When $l=\min\{m,l,n\}$, we make no change of $\mathcal{A}$ and $\mathcal{B}$ (lines 11 - 15). After initializing several relevant variables (lines 16 - 18), Algorithm~\ref{alg:HEGMM-enhanced} goes through a loop to compute and accumulate the partial products (lines 19-28). To be more specific, we first conduct $\epsilon$ and $\omega$ transformations based on the expanded matrix ($\mathcal{A}$ or $\mathcal{B}$) (lines 20 - 21), which are combined together into $\mathcal{C}_{temp}$ using the element-wise HE multiplication (line 22). The algorithm then extracts the possible $t$ copies of  
$\epsilon_{m\times n}(\sigma(\mathcal{A}))\odot \omega_{m\times n}(\tau(\mathcal{B}))$ from $\mathcal{C}_{temp}$ and accumulates them to $\mathcal{C}_{m\times n}$, according to Theorems~\ref{thrm:dupA} and ~\ref{thrm:dupB}, and the redundant copies are excluded from the $\mathcal{C}_{m\times n}$. 

Note that, compared with Algorithm~\ref{alg:HEGMM}, Algorithm~\ref{alg:HEGMM-enhanced} only needs to perform $p=\min\{m,l,n\}$ loops (line 19) instead of $l$. We assume that the proper order is adopted when flattening the matrix: When $p = \min\{m,l,n\} = m$, ${\mathcal{A}}$ is expanded and the column-major order is adopted to flatten matrices; When $p = \min\{m,l,n\} = n$, ${\mathcal{B}}$ is expanded and the row-major order is adopted to flatten matrices; When $p = \min\{m,l,n\} = l$, neither ${\mathcal{A}}$ and ${\mathcal{B}}$ is expanded, and either major order can be adopted to flatten matrices. 

\else

\subsubsection{Algorithm Details}
\begin{center}
\begin{minipage}{.9\linewidth}
\begin{algorithm}[H]
\caption{HEGMM-Enhanced}
\LinesNumbered
\KwIn{matrix $\mathcal{A}_{m\times l}$, $\mathcal{B}_{l\times n}$}
\KwOut{matrix $\mathcal{C}_{m\times n} =  \mathcal{A}_{m\times l} \times \mathcal{B}_{l\times n}$}

\setcounter{AlgoLine}{0}
\textbf{[Step0-1]}

$p\gets \min(m,l,n)$

$t\gets \left \lceil l/p \right \rceil$

\textbf{[Step0-2]}

\uIf{p=m}
    {
        $\Bar{\mathcal{A}}_{M\times N} \gets \underset{t}{[\underbrace{\mathcal{A};\mathcal{A};\mathcal{A};...;\mathcal{A}}]^T}$
        
        $\Bar{\mathcal{B}}_{M\times N} \gets \mathcal{B}$
    }
\uElseIf{p=n}
    {
        $\Bar{\mathcal{A}}_{M\times N} \gets \mathcal{A}$
        
        $\Bar{\mathcal{B}}_{M\times N} \gets \underset{t}{[\underbrace{\mathcal{B};\mathcal{B};\mathcal{B};...;\mathcal{B}}]}$
    }
\Else
    {
        $\Bar{\mathcal{A}}_{M\times N} \gets \mathcal{A}$
        
        $\Bar{\mathcal{B}}_{M\times N} \gets \mathcal{B}$ 
    }

\textbf{[Step1]}

$ct.\Bar{\mathcal{A}}^{(0)}  \gets \epsilon^0_{M\times \max(l,N)}(\sigma(\Bar{\mathcal{A}}))$ //dimension? 

$ct.\Bar{\mathcal{B}}^{(0)}  \gets \omega_{\max(l,M)\times N}^0(\tau(\Bar{\mathcal{B}}))$ //dimensinos?

\textbf{[Step2]}

\For{k=0,1,...,$p-1$}
    {
        $ct.\Bar{\mathcal{A}}^{(k)} \gets \epsilon_{M\times N}^k(ct.\Bar{\mathcal{A}}^{(0)})$ //dimensions?
        
        $ct.\Bar{\mathcal{B}}^{(k)} \gets \omega_{M\times N}^k(ct.\Bar{\mathcal{B}}^{(0)})$ //dimensinos?
        
        \uIf{$t\cdot (k+1) \le l$}
        {
            $ct.{\mathcal{C}} \gets ct.{\mathcal{C}}+ ct.\Bar{\mathcal{A}}^{(k)} \odot ct.\Bar{\mathcal{B}}^{(k)}$
        }
        \Else
        {
           $mask=\underset{(t-1)\cdot m\cdot n}{[\underbrace{1,1,...,1}},\underset{m\cdot n}{\underbrace{0,...,0}]}$
        
            $ct.{\mathcal{C}} \gets ct.{\mathcal{C}}+ ct.\Bar{\mathcal{A}}^{(k)} \odot ct.\Bar{\mathcal{B}}^{(k)} \odot mask$
        }
    }

\For{i=0,1,...,$t-1$}
    {
        $ct.\mathcal{C} \gets ct.\mathcal{C}+He\text{-}Rot(ct.\mathcal{C}, i\cdot p)$
    }

\textbf{[Step3]}

$\mathcal{C}_{m\times n} \gets ct.\mathcal{C}$

\Return{$\mathcal{C}_{m\times n}$}
\label{alg:HEGMM-enhanced}
\end{algorithm} 
\end{minipage}
\end{center}

The core idea is compromising space for time. However, the space is limited. The number of slots in a ciphertext ($len$) must be no smaller than the largest number of slots required by $\mathcal{A}_{m\times l}$, $\mathcal{B}_{l\times n}$, and $\mathcal{C}_{m\times n}$. That is,
\begin{equation*}
\label{eqn:ciphsize}
    len \geq
    \begin{cases}
      \max(m,l,n)\cdot m\cdot {\left \lceil l/m \right \rceil} & \text{ if }m=\min(m,l,n)\\
      \max(m,l,n)\cdot n\cdot {\left \lceil l/n \right \rceil} & \text{ if }n=\min(m,l,n)\\
      m\cdot n & \text{ if }l=\min(m,l,n)
    \end{cases}
\end{equation*}

\fi

\section{Experiments}

In this section, we evaluate the performance of the two algorithms developed in this paper, i.e., HEGMM and HEGMM-Enhanced, and compare them with the state-of-the-art schemes for HE-based matrix multiplication.

\subsection{Experimental platform}
\label{Exp-setup}
We implemented HEGMM and HEGMM-Enhanced using a Python HE library, named Pyfhel~\cite{ibarrondo2021pyfhel} with BFV scheme \cite{Bfv12,fan2012BFV}. We set the HE scheme based on the RLWE (Ring Learning With Errors)~\cite{lyubashevsky2010ideal} assumption over the cyclotomic ring $R_q = \mathbb{Z}_q[X]/(X^N + 1)$ with $N=2^{12}$. Thus each ciphertext can hold up to $N= 2^{12}$ slots for plaintext values, the largest square matrix that can be accommodated in one ciphertext is thus $64 \times 64$. 

In our experiments, we studied the following approaches. 
\begin{itemize}
\item \textbf{E2DM-S}, which is presented in~\cite{jiang2018secure} on square matrix multiplication. For a general MM $\mathcal{A}_{m\times l} \times \mathcal{B}_{l\times n}$, we can transform $\mathcal{A}_{m\times l}$ and $\mathcal{B}_{l\times n}$ to two square matrices, $\mathcal{A'}_{d\times d}$ and $\mathcal{B'}_{d\times d}$ with $d=\max\{m,l,n\}$ and use this algorithm to calculate the result;   

\item \textbf{E2DM-R}, which is presented in~\cite{jiang2018secure} on rectangular matrix multiplication $\mathcal{A}_{r\times d} \times \mathcal{B}_{d\times d}$. For a general MM $\mathcal{A}_{m\times l} \times \mathcal{B}_{l\times n}$, we can expand $\mathcal{A}_{m\times l}$ and/or $\mathcal{B}_{l\times n}$ accordingly and use this algorithm to calculate the result;

\item \textbf{Huang-MM}, which is introduced in \cite{huang2022secure} and implemented with Pyfhel \cite{ibarrondo2021pyfhel}. 

\item \textbf{HEGMM}, which is shown in Algorithm~\ref{alg:HEGMM}. 

\item \textbf{HEGMM-En}, which is shown in Algorithm~\ref{alg:HEGMM-enhanced}. 

\end{itemize}

We randomly generated 2000 pairs of matrices, with column and row numbers evenly distributed with $[1,64]$, as the test cases. Note that, even though Huang \textit{et al.} \cite{huang2022secure}, HEGMM and HEGMM-Enhanced can handle MM with column or row numbers exceeding 64, as long as the total element is no more than $2^{12}$, we limited the largest dimension size to 64 so that \textbf{E2DM-S} and \textbf{E2DM-R} can always apply. We assume that $\sigma$ and $\tau$ transformations of E2DM and HEGMM are performed on plaintext, and to be fair, we assume the portion of Huang \textit{et al.}'s algorithm, i.e., extracting diagonal vector from matrix, is also performed on plaintext. All experiments were conducted on a server with Intel Xeon Silver 4114 with 10 cores at 2.2GHz.

\subsection{Computational time evaluations}

To better understand the performance of the five different algorithms listed above, we categorize the test cases into 5 groups: (1) $m=\min\{m,l,n\}$; (2) $l=\min\{m,l,n\}$; (3) $n=\min\{m,l,n\}$; (4) $l \bmod m =0$; (5) $m=l=n$ (the square matrix). Note that cases in (4) and (5) are the most favorable ones for \textbf{E2DM-R} and \textbf{E2DM-S}, respectively. For test cases in each group, execution times were collected for the five different approaches. We use the better ones by \textbf{E2DM-S} and \textbf{E2DM-R} as the performance that can achieved by \textbf{E2DM}. We then calculate the speedups that can be achieved using \textbf{HEGMM}, \textbf{HEGMM-En} over \textbf{E2DM} and \textbf{Huang \textit{et al.}}. 

The average, median, and maximum speedup for each group, as well as the overall results, are listed in Tables~\ref{tab:hegmmavgmed} and ~\ref{tab:hegmmenavgmed}.  

\begin{table*}[htbp]
  \centering
  \begin{threeparttable}
  \caption{The performance comparison of HEGMM, E2DM \cite{jiang2018secure} and Huang \textit{et al.} \cite{huang2022secure} in different scenarios.}
    \begin{tabular}{c|rr|rr|rr|rr|rr|rr}
     & \multicolumn{2}{c|}{$m=min(m,l,n)$} & \multicolumn{2}{c|}{$l=min(m,l,n)$} & \multicolumn{2}{c|}{$n=min(m,l,n)$} & \multicolumn{2}{c|}{$l \mod m=0$} & \multicolumn{2}{c|}{$square$} & \multicolumn{2}{c}{overall} \\
         & \multicolumn{1}{c}{s-up1$^\dagger$} & \multicolumn{1}{c|}{s-up2} & \multicolumn{1}{c}{s-up1} & \multicolumn{1}{c|}{s-up2} & \multicolumn{1}{c}{s-up1} & \multicolumn{1}{c|}{s-up2} & \multicolumn{1}{c}{s-up1} & \multicolumn{1}{c|}{s-up2} & \multicolumn{1}{c}{s-up1} & \multicolumn{1}{c|}{s-up2} & \multicolumn{1}{c}{s-up1} & \multicolumn{1}{c}{s-up2}\\
    \hline
    Average & 0.90 & 2.21 & 10.74 & 1.99 & 1.83 & 2.29 & 1.49 & 2.12 & 1.00 & 1.67 & 3.30 & 1.93\\
    Median & 0.58 & 2.21 & 3.61 & 1.95 & 1.20 & 2.35 & 0.52 & 2.10 & 1.00 & 1.72 & 1.04 & 2.04 \\
    Max  & 39.50 & 3.25 & 154.12 & 4.96 & 136.82 & 3.28 & 39.50 & 3.07 & 1.02 & 2.36 & 154.12 & 4.96 \\
    \end{tabular}%
    
  \label{tab:hegmmavgmed}%
  \begin{tablenotes}
    \footnotesize
    \item$[\dagger] \emph{s-up1}$ is the speedup achieved by 
    \textbf{HEGMM} over the best of \textbf{E2DM} \cite{jiang2018secure}; \emph{s-up2} is the speedup achieved by \textbf{HEGMM} over \textbf{Huang \textit{et al.} \cite{huang2022secure}};
  \end{tablenotes}
  
  \end{threeparttable}
\end{table*}%

\begin{table*}[htbp]
  \centering
  \begin{threeparttable}
  \caption{The performance comparison of HEGMM-En, E2DM \cite{jiang2018secure} and Huang \textit{et al.} \cite{huang2022secure} in different scenarios.}
    \begin{tabular}{c|rr|rr|rr|rr|rr|rr}
     & \multicolumn{2}{c|}{$m=min(m,l,n)$} & \multicolumn{2}{c|}{$l=min(m,l,n)$} & \multicolumn{2}{c|}{$n=min(m,l,n)$} & \multicolumn{2}{c|}{$l \mod m=0$} & \multicolumn{2}{c|}{$square$} & \multicolumn{2}{c}{overall} \\
         & \multicolumn{1}{c}{s-up3$^\ddagger$} & \multicolumn{1}{c|}{s-up4} & \multicolumn{1}{c}{s-up3} & \multicolumn{1}{c|}{s-up4} & \multicolumn{1}{c}{s-up3} & \multicolumn{1}{c|}{s-up4} & \multicolumn{1}{c}{s-up3} & \multicolumn{1}{c|}{s-up4} & \multicolumn{1}{c}{s-up3} & \multicolumn{1}{c|}{s-up4} & \multicolumn{1}{c}{s-up3} & \multicolumn{1}{c}{s-up4}\\
    \hline
    Average & 1.69 & 6.60 & 10.32 & 2.01 & 4.06 & 6.55 & 2.74 & 7.28 & 0.99 & 1.66 & 4.13 & 4.50\\
    Median & 1.13 & 4.88 & 3.70 & 1.97 & 2.56 & 4.83 & 1.37 & 6.23 & 1.00 & 1.72 & 1.38 & 2.48 \\
    Max  & 33.28 & 23.29 & 132.42 & 4.26 & 113.31 & 23.68 & 33.28 & 21.75 & 1.01 & 2.36 & 132.42 & 23.68 \\
    \end{tabular}%
    
  \label{tab:hegmmenavgmed}%
  \begin{tablenotes}
    \footnotesize
    \item$[\ddagger] \emph{s-up3}$ is the speedup achieved by \textbf{HEGMM-En} over the best of \textbf{E2DM} \cite{jiang2018secure}; \emph{s-up4} is the speedup achieved by \textbf{HEGMM-En} over \textbf{Huang \textit{et al.} \cite{huang2022secure}};
  \end{tablenotes}
  
  \end{threeparttable}
\end{table*}%

As in Table~\ref{tab:hegmmavgmed}, \textbf{HEGMM} outperforms \textbf{Huang \textit{el al.}} in all groups, with a speedup of $1.93$ on average and the maximum of over $4.96$. Compared with \textbf{E2DM}, \textbf{HEGMM} can achieve better performance in all cases other than if the matrices are square or when $m=min(m,l,n)$. As shown in Table~\ref{tab:hegmmavgmed}, \textbf{HEGMM} can achieve a speedup of $3.3$ on average with the maximum of over $154.12$ over the best of \textbf{E2DM}. 
When source matrices are square, \textbf{HEGMM} is equivalent to \textbf{E2DM} with slight overhead for taking care of generality of matrices. When $m=min(m,l,n)$, the time complexity of \textbf{E2DM-R} is $\mathcal{O}(m)$ while \textbf{HEGMM} is $\mathcal{O}(l)$. Therefore, \textbf{E2DM-R} can potentially achieve better performance, especially when $m <<l$. 

The enhanced algorithm, i.e., \textbf{HEGMM-En}, can significantly outperform the rest of the approaches for arbitrary HE MM, as shown in Table~\ref{tab:hegmmenavgmed}. This is because 
\textbf{HEGMM-En} can reduce HE-Mult operations significantly by properly duplicating the source matrices. Specifically, \textbf{HEGMM-En} can achieve an average speedup of $4.13$ with the maximum of $132.42$ over best of \textbf{E2DM}, and an average speedup of $4.50$ with the maximum of $23.68$ over the \textbf{Huang \textit{et al.}}. For square matrices, \textbf{HEGMM-En} is equivalent to \textbf{E2DM} and requires slightly more time than due to the overhead for taking care of the generality of matrices. 

We also use Figure~\ref{Fig.expbar} to compare the performance of these approaches from a different perspective. Specifically, Figure~\ref{Fig.expbar} shows the number of test cases that can achieve speedups between $(0,1]$, $[1,2]$, and $(2, +\infty)$ by \textbf{HEGMM}, \textbf{HEGMM-En} and \textbf{Huang \textit{et al.}} over the best results by \textbf{E2DM-S} and \textbf{E2DM-R}. In a total of 2000 test cases, there were 1324 cases that \textbf{HEGMM} outperform both \textbf{E2DM-S} and \textbf{E2DM-R}, while it is 1805 for \textbf{HEGMM-En}, which indicates that \textbf{HEGMM-En} performs significantly better than \textbf{HEGMM}. For \textbf{Huang \textit{et al.}}, only 610 samples outperform \textbf{E2DM}. 
Overall, the experimental findings indicate that the algorithms \textbf{HEGMM} and \textbf{HEGMM-EN} exhibit a significant performance superiority compared to current methodologies in 66.2\% and 90.2\% of the samples, respectively.

\begin{figure}[h]
\centering
\includegraphics[width=0.40\textwidth]{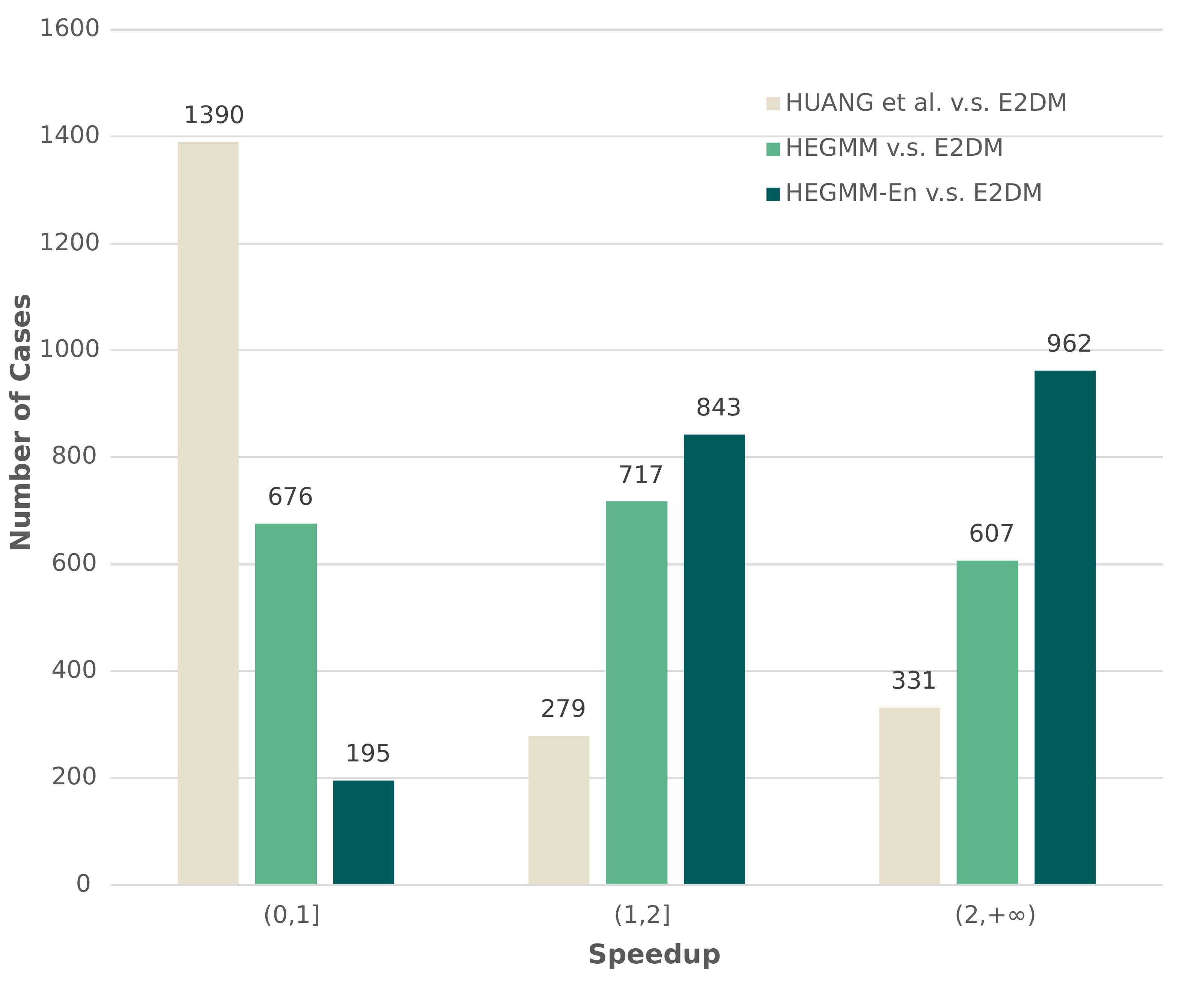}
\caption{The statistics of the speedups for the algorithms HEGMM, HEGMM-En, E2DM \cite{jiang2018secure}, and Huang \textit{et al.} \cite{huang2022secure}.}
\label{Fig.expbar}
\end{figure}

\begin{figure}[h]
\centering
\includegraphics[width=0.40\textwidth]{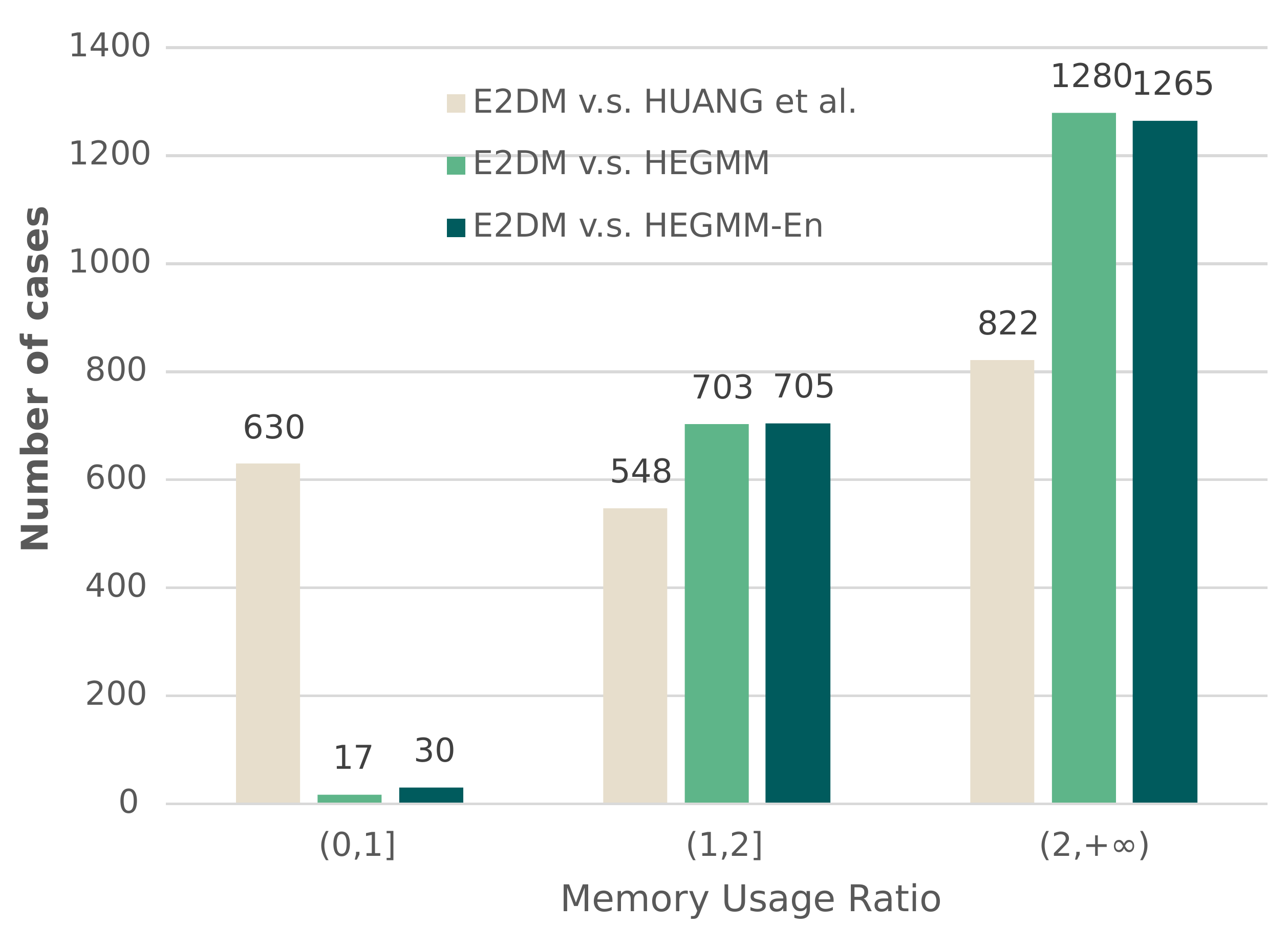}
\caption{The statistics of memory usage ratio for the algorithms HEGMM, HEGMM-En, E2DM \cite{jiang2018secure}, and Huang \textit{et al.} \cite{huang2022secure}.}

\label{Fig.membar}
\end{figure}

\subsection{Memory evaluations}

HE computations may demand not only excessive computation time but also memory usage as well. We are therefore interested in studying the memory usage of these approaches. We collected the memory usage for each algorithm during its runtime for our test cases with results normalized against the memory usage by \textbf{E2DM} and presented in Figure~\ref{Fig.membar}, where a total of 2000 experimental sets were conducted. In comparison to \textbf{E2DM}, both \textbf{HEGMM} and \textbf{HEGMM-En} tend to consume less memory. As shown in Figure~\ref{Fig.membar}, less than 17 (resp. 30) out of the total 2000 test results show that \textbf{E2DM} consumes less memory than \textbf{HEGMM} (resp. \textbf{HEGMM-EN}). In contrast, 630 test cases using \textbf{Huang \textit{et al.}} have higher memory usage compared to \textbf{E2DM}. Overall, the experimental results clearly demonstrate 
the advantage of memory usage efficiency
of \textbf{HEGMM} and \textbf{HEGMM-En} over the existing approaches.

\subsection{Evaluation of large matrix multiplication}

Our test cases above are limited to the maximum matrix dimension of 64x64, the largest one that can fit into one ciphertext in our setting. When matrix sizes exceed this limit, we can resort to the traditional \emph{blocking} algorithm, i.e., by dividing a large matrix into a series of smaller blocks, to perform the MM calculation. We want to study the performance of our proposed approaches when incorporated into MM blocking algorithms
for $\mathcal{A}_{100\times 100}\times \mathcal{B}_{100\times 100}$.

Partitioning large source matrices properly based on different MM algorithms is an interesting problem but beyond the scope of this paper. In our experiments, we hire two intuitive partition methods: \textbf{P1:} partitioning the matrix ${100\times 100}$ to four equal-size square matrices of ${50\times 50}$; \textbf{P2:} partitioning the matrix ${100\times 100}$ to four sub matrices of ${64\times 64}$, ${64\times 36}$, ${36\times 64}$, and ${36\times 36}$.  

Different HE MM algorithms were employed for blocking MMs. We ran the experiments 10 times, and the average results were collected and shown in Table~\ref{tab:blcres}. As expected, for \textbf{P1} when all matrices are square, \textbf{E2DM-S}, \textbf{E2DM-R}, \textbf{HEGMM} and \textbf{HEGMM-EN} perform quite similarly, while \textbf{HEGMM} and \textbf{HEGMM-EN} take a little longer due to overhead in dealing with the generality of the matrices. \textbf{Huang \textit{et al.}} shows a much slower performance than the others. We believe this is because that \textbf{Huang \textit{et al.}} approach requires duplicating diagonals of a source matrix with the complexity of $O(logN)$, with $N$ the size of the matrix. The duplication operation involves expensive \emph{HE-CMult} and \emph{HE-Rot} operations. This is particularly computationally expensive when $N$ is not a power of two. In contrast, the time complexity of same step in \textbf{E2DM} and \textbf{HEGMM} is $O(2)$ for \textbf{P1}.

For $\textbf{P2}$, \textbf{HEGMM}, \textbf{HEGMM-EN}, and \textbf{Huang \textit{et al.}} can perform better because they can take advantage of the irregular shapes of the matrices. In particular, \textbf{HEGMM-EN} (resp. \textbf{HEGMM}) has a complexity of $O(\min(m,l,n))$ (resp. $O(l)$). In contrast, \textbf{E2DM-S} runs much longer because it needs to expand matrices ${64\times 36}$ and ${36\times 64}$ to form ${64\times 64}$ matrix. 
\textbf{E2DM-R} is incapable of processing matrices with such irregular shapes, as it has a tendency to enlarge matrix of $36 \times 64$ to $72 \times 72$, which is larger than the ciphertext size.

\setlength{\tabcolsep}{1mm}{
\begin{table}[htbp]
  \centering
  \caption{Time evaluation of the blocking algorithm}
    \begin{tabular}{c|ccccc}
    \textbf{Partition} & \multicolumn{1}{c}{\textbf{E2DM-S}} & \multicolumn{1}{c}{\textbf{E2DM-R}} & \multicolumn{1}{c}{\textbf{Huang \textit{et al.}}} & \multicolumn{1}{c}{\textbf{HEGMM}} & \multicolumn{1}{c}{\textbf{HEGMM-EN}}\\
    \hline
    \textbf{P1}   & 39.06s & \textbf{39.01s} & 74.34s & 39.12s & 39.15s\\
    \textbf{P2}   & 29.76s & N/A  & 37.51s & \textbf{26.17s} & 26.23s \\
    \end{tabular}%
  \label{tab:blcres}%
\end{table}%
}

\section{Conclusions}
HE has great potential for security and privacy protection when outsourcing data processing to the cloud. However, the excessive computational overhead associated with the HE operations makes it prohibitive for many practical cloud applications. We study how to reduce the HE computational cost for general MM operation, an essential building block in many computational fields. We present two HE MM algorithms, with one improving another, to reduce the computational complexity of MM by taking advantage of the SIMD structure in the HE scheme. We also conduct rigorous analytical studies on the correctness and computational complexity of these two algorithms. Experiment results show that our proposed approach can significantly outperform the existing methods. In our future research, we plan to investigate how to reduce the HE computational cost for sparse matrix multiplication.

\section*{Acknowledgements}
This work was supported in part by the Air Force Office of Scientific Research (AFOSR) and the Air Force Research Laboratory / Information Directorate (AFRL/RI), Rome, NY under the 2021 Summer Faculty Fellowship Program, and Information Directorate Internship Program, respectively.  The views and conclusions contained herein are those of the authors and should not be interpreted as necessarily representing the official policies or endorsements, either expressed or implied, of the Air Force Research Laboratory or the U.S. Government. Approved for Public Release on 06 Mar 2024. Distribution is Unlimited. Case Number: 2024-0184  (original case number(s):  AFRL-2024-0944

\bibliographystyle{abbrv}
\bibliography{refer.bib}

\begin{thebibliography}{10}
\providecommand{\url}[1]{#1}
\csname url@samestyle\endcsname
\providecommand{\newblock}{\relax}
\providecommand{\bibinfo}[2]{#2}
\providecommand{\BIBentrySTDinterwordspacing}{\spaceskip=0pt\relax}
\providecommand{\BIBentryALTinterwordstretchfactor}{4}
\providecommand{\BIBentryALTinterwordspacing}{\spaceskip=\fontdimen2\font plus
\BIBentryALTinterwordstretchfactor\fontdimen3\font minus \fontdimen4\font\relax}
\providecommand{\BIBforeignlanguage}[2]{{%
\expandafter\ifx\csname l@#1\endcsname\relax
\typeout{** WARNING: IEEEtran.bst: No hyphenation pattern has been}%
\typeout{** loaded for the language `#1'. Using the pattern for}%
\typeout{** the default language instead.}%
\else
\language=\csname l@#1\endcsname
\fi
#2}}
\providecommand{\BIBdecl}{\relax}
\BIBdecl

\bibitem{Duong2017TMMP}
\BIBentryALTinterwordspacing
D.~H. Duong, P.~K. Mishra, and M.~Yasuda, ``Efficient secure matrix multiplication over lwe-based homomorphic encryption,'' \emph{Tatra Mountains Mathematical Publications}, vol.~67, no.~1, pp. 69--83, 2017. [Online]. Available: \url{https://doi.org/10.1515/tmmp-2016-0031}
\BIBentrySTDinterwordspacing

\bibitem{Mishra2017ISPE}
P.~K. Mishra, D.~H. Duong, and M.~Yasuda, ``Enhancement for secure multiple matrix multiplications over ring-lwe homomorphic encryption,'' in \emph{Information Security Practice and Experience}, J.~K. Liu and P.~Samarati, Eds.\hskip 1em plus 0.5em minus 0.4em\relax Springer, 2017, pp. 320--330.

\bibitem{varghese2018next}
B.~Varghese and R.~Buyya, ``Next generation cloud computing: New trends and research directions,'' \emph{Future Generation Computer Systems}, vol.~79, pp. 849--861, 2018.

\bibitem{cheon2018multi}
J.~H. Cheon, A.~Kim, and D.~Yhee, ``Multi-dimensional packing for heaan for approximate matrix arithmetics,'' \emph{Cryptology ePrint Archive}, 2018.

\bibitem{atallah2002secure}
M.~J. Atallah, K.~N. Pantazopoulos, J.~R. Rice, and E.~E. Spafford, ``Secure outsourcing of scientific computations,'' in \emph{Advances in Computers}.\hskip 1em plus 0.5em minus 0.4em\relax Elsevier, 2002, vol.~54, pp. 215--272.

\bibitem{lei2014achieving}
X.~Lei, X.~Liao, T.~Huang, and F.~Heriniaina, ``Achieving security, robust cheating resistance, and high-efficiency for outsourcing large matrix multiplication computation to a malicious cloud,'' \emph{Information sciences}, vol. 280, pp. 205--217, 2014.

\bibitem{fu2017secure}
S.~Fu, Y.~Yu, and M.~Xu, ``A secure algorithm for outsourcing matrix multiplication computation in the cloud,'' in \emph{Proceedings of the Fifth ACM international workshop on security in cloud computing}, 2017, pp. 27--33.

\bibitem{zhang2019practical}
S.~Zhang, C.~Tian, H.~Zhang, J.~Yu, and F.~Li, ``Practical and secure outsourcing algorithms of matrix operations based on a novel matrix encryption method,'' \emph{IEEE Access}, vol.~7, pp. 53\,823--53\,838, 2019.

\bibitem{mishra2021fast}
P.~K. Mishra, D.~Rathee, D.~H. Duong, and M.~Yasuda, ``Fast secure matrix multiplications over ring-based homomorphic encryption,'' \emph{Information Security Journal: A Global Perspective}, vol.~30, no.~4, pp. 219--234, 2021.

\bibitem{wang2019secure}
S.~Wang and H.~Huang, ``Secure outsourced computation of multiple matrix multiplication based on fully homomorphic encryption,'' \emph{KSII Transactions on Internet and Information Systems (TIIS)}, vol.~13, no.~11, pp. 5616--5630, 2019.

\bibitem{cheon2018homomorphic}
J.~H. Cheon and A.~Kim, ``Homomorphic encryption for approximate matrix arithmetic,'' \emph{Cryptology ePrint Archive}, 2018.

\bibitem{tian2014somewhat}
Y.~Tian, M.~Al-Rodhaan, B.~Song, A.~Al-Dhelaan, and T.~H. Ma, ``Somewhat homomorphic cryptography for matrix multiplication using gpu acceleration,'' in \emph{2014 International Symposium on Biometrics and Security Technologies (ISBAST)}.\hskip 1em plus 0.5em minus 0.4em\relax IEEE, 2014, pp. 166--170.

\bibitem{hesamifard2018privacy}
E.~Hesamifard, H.~Takabi, M.~Ghasemi, and R.~N. Wright, ``Privacy-preserving machine learning as a service.'' \emph{Proc. Priv. Enhancing Technol.}, vol. 2018, no.~3, pp. 123--142, 2018.

\bibitem{hiromasa2016packing}
R.~Hiromasa, M.~Abe, and T.~Okamoto, ``Packing messages and optimizing bootstrapping in gsw-fhe,'' \emph{IEICE TRANSACTIONS on Fundamentals of Electronics, Communications and Computer Sciences}, vol.~99, no.~1, pp. 73--82, 2016.

\bibitem{scale2015state}
R.~Scale, ``State of the cloud report,'' Tech. Rep, Tech. Rep., 2015.

\bibitem{kushilevitz1997replication}
E.~Kushilevitz and R.~Ostrovsky, ``Replication is not needed: Single database, computationally-private information retrieval,'' in \emph{Proceedings 38th annual symposium on foundations of computer science}.\hskip 1em plus 0.5em minus 0.4em\relax IEEE, 1997, pp. 364--373.

\bibitem{benaloh1987verifiable}
J.~D.~C. Benaloh, \emph{Verifiable secret-ballot elections}.\hskip 1em plus 0.5em minus 0.4em\relax Yale University, 1987.

\bibitem{rivest1978method}
R.~L. Rivest, A.~Shamir, and L.~Adleman, ``A method for obtaining digital signatures and public-key cryptosystems,'' \emph{Communications of the ACM}, vol.~21, no.~2, pp. 120--126, 1978.

\bibitem{goldwasser1982probabilistic}
S.~Goldwasser and S.~Micali, ``Probabilistic encryption \& how to play mental poker keeping secret all partial information,'' in \emph{Proceedings of the fourteenth annual ACM symposium on Theory of computing}, 1982, pp. 365--377.

\bibitem{elgamal1985public}
T.~ElGamal, ``A public key cryptosystem and a signature scheme based on discrete logarithms,'' \emph{IEEE transactions on information theory}, vol.~31, no.~4, pp. 469--472, 1985.

\bibitem{benaloh1994receipt}
J.~Benaloh and D.~Tuinstra, ``Receipt-free secret-ballot elections,'' in \emph{Proceedings of the twenty-sixth annual ACM symposium on Theory of computing}, 1994, pp. 544--553.

\bibitem{naccache1998new}
D.~Naccache and J.~Stern, ``A new public key cryptosystem based on higher residues,'' in \emph{Proceedings of the 5th ACM Conference on Computer and Communications Security}, 1998, pp. 59--66.

\bibitem{okamoto1998new}
T.~Okamoto and S.~Uchiyama, ``A new public-key cryptosystem as secure as factoring,'' in \emph{International conference on the theory and applications of cryptographic techniques}.\hskip 1em plus 0.5em minus 0.4em\relax Springer, 1998, pp. 308--318.

\bibitem{paillier1999public}
P.~Paillier, ``Public-key cryptosystems based on composite degree residuosity classes,'' in \emph{International conference on the theory and applications of cryptographic techniques}.\hskip 1em plus 0.5em minus 0.4em\relax Springer, 1999, pp. 223--238.

\bibitem{damgaard2001generalisation}
I.~Damg{\aa}rd and M.~Jurik, ``A generalisation, a simplification and some applications of paillier's probabilistic public-key system,'' in \emph{International workshop on public key cryptography}.\hskip 1em plus 0.5em minus 0.4em\relax Springer, 2001, pp. 119--136.

\bibitem{kawachi2007multi}
A.~Kawachi, K.~Tanaka, and K.~Xagawa, ``Multi-bit cryptosystems based on lattice problems,'' in \emph{International Workshop on Public Key Cryptography}.\hskip 1em plus 0.5em minus 0.4em\relax Springer, 2007, pp. 315--329.

\bibitem{tackmannconstructing}
S.~C. U. M.~B. Tackmann, ``Constructing confidential channels from authenticated channels—public-key encryption revisited.''

\bibitem{boneh2005evaluating}
D.~Boneh, E.-J. Goh, and K.~Nissim, ``Evaluating 2-dnf formulas on ciphertexts,'' in \emph{Theory of cryptography conference}.\hskip 1em plus 0.5em minus 0.4em\relax Springer, 2005, pp. 325--341.

\bibitem{sander1999non}
T.~Sander, A.~Young, and M.~Yung, ``Non-interactive cryptocomputing for nc/sup 1,'' in \emph{40th Annual Symposium on Foundations of Computer Science (Cat. No. 99CB37039)}.\hskip 1em plus 0.5em minus 0.4em\relax IEEE, 1999, pp. 554--566.

\bibitem{lopez2012fly}
A.~L{\'o}pez-Alt, E.~Tromer, and V.~Vaikuntanathan, ``On-the-fly multiparty computation on the cloud via multikey fully homomorphic encryption,'' in \emph{Proceedings of the forty-fourth annual ACM symposium on Theory of computing}, 2012, pp. 1219--1234.

\bibitem{brakerski2011fully}
Z.~Brakerski and V.~Vaikuntanathan, ``Fully homomorphic encryption from ring-lwe and security for key dependent messages,'' in \emph{Annual cryptology conference}.\hskip 1em plus 0.5em minus 0.4em\relax Springer, 2011, pp. 505--524.

\bibitem{ames2015secure}
S.~Ames, M.~Venkitasubramaniam, A.~Page, O.~Kocabas, and T.~Soyata, ``Secure health monitoring in the cloud using homomorphic encryption: A branching-program formulation,'' in \emph{Enabling Real-Time Mobile Cloud Computing through Emerging Technologies}.

\bibitem{lyubashevsky2010ideal}
V.~Lyubashevsky, C.~Peikert, and O.~Regev, ``On ideal lattices and learning with errors over rings,'' in \emph{29th Intl. Conference on the Theory and Applications of Cryptographic Techniques}.\hskip 1em plus 0.5em minus 0.4em\relax Springer, 2010, pp. 1--23.

\bibitem{reagen2021cheetah}
B.~Reagen, W.-S. Choi, Y.~Ko, V.~T. Lee, H.-H.~S. Lee, G.-Y. Wei, and D.~Brooks, ``Cheetah: Optimizing and accelerating homomorphic encryption for private inference,'' in \emph{IEEE International Symposium on High-Performance Computer Architecture (HPCA)}.

\bibitem{nocker2023he}
M.~Nocker, D.~Drexel, M.~Rader, A.~Montuoro, and P.~Sch{\"o}ttle, ``He-man--homomorphically encrypted machine learning with onnx models,'' \emph{arXiv preprint arXiv:2302.08260}, 2023.

\bibitem{van2010fully}
M.~Van~Dijk, C.~Gentry, S.~Halevi, and V.~Vaikuntanathan, ``Fully homomorphic encryption over the integers,'' in \emph{Annual international conference on the theory and applications of cryptographic techniques}.\hskip 1em plus 0.5em minus 0.4em\relax Springer, 2010, pp. 24--43.

\bibitem{ishai2007evaluating}
Y.~Ishai and A.~Paskin, ``Evaluating branching programs on encrypted data,'' in \emph{Theory of Cryptography Conference}.\hskip 1em plus 0.5em minus 0.4em\relax Springer, 2007, pp. 575--594.

\bibitem{acar2018survey}
A.~Acar, H.~Aksu, A.~S. Uluagac, and M.~Conti, ``A survey on homomorphic encryption schemes: Theory and implementation,'' \emph{ACM Computing Surveys (Csur)}, vol.~51, no.~4, pp. 1--35, 2018.

\bibitem{ghobaei2018autonomic}
M.~Ghobaei-Arani, S.~Jabbehdari, and M.~A. Pourmina, ``An autonomic resource provisioning approach for service-based cloud applications: A hybrid approach,'' \emph{Future Generation Computer Systems}, vol.~78, pp. 191--210, 2018.

\bibitem{pancholi2016enhancement}
V.~R. Pancholi and B.~P. Patel, ``Enhancement of cloud computing security with secure data storage using aes,'' \emph{International Journal for Innovative Research in Science and Technology}, vol.~2, no.~9, pp. 18--21, 2016.

\bibitem{rajaraman2014cloud}
V.~Rajaraman, ``Cloud computing,'' \emph{Resonance}, vol.~19, no.~3, pp. 242--258, 2014.

\bibitem{power2018revenue}
B.~Power and J.~Weinman, ``Revenue growth is the primary benefit of the cloud,'' \emph{IEEE Cloud Computing}, vol.~5, no.~4, pp. 89--94, 2018.

\bibitem{becker2017cloudscale}
S.~Becker, G.~Brataas, M.~Cecowski, D.~Huljeni{\'c}, S.~Lehrig, and I.~Stupar, ``The cloudscale method for managers,'' in \emph{Engineering Scalable, Elastic, and Cost-Efficient Cloud Computing Applications}.\hskip 1em plus 0.5em minus 0.4em\relax Springer, 2017, pp. 149--165.

\bibitem{fawzi2022discovering}
A.~Fawzi, M.~Balog, A.~Huang, T.~Hubert, B.~Romera-Paredes, M.~Barekatain, A.~Novikov, F.~J. R~Ruiz, J.~Schrittwieser, G.~Swirszcz \emph{et~al.}, ``Discovering faster matrix multiplication algorithms with reinforcement learning,'' \emph{Nature}, vol. 610, no. 7930, pp. 47--53, 2022.

\bibitem{liu2011nist}
F.~Liu, J.~Tong, J.~Mao, R.~Bohn, J.~Messina, L.~Badger, D.~Leaf \emph{et~al.}, ``Nist cloud computing reference architecture,'' \emph{NIST special publication}, vol. 500, no. 2011, pp. 1--28, 2011.

\bibitem{jiang2020novel}
P.~Jiang, C.~Hong, and G.~Agrawal, ``A novel data transformation and execution strategy for accelerating sparse matrix multiplication on gpus,'' in \emph{Proceedings of the 25th ACM SIGPLAN symposium on principles and practice of parallel programming}, 2020, pp. 376--388.

\bibitem{8374488}
P.~Valero-Lara, I.~Martínez-Pérez, S.~Mateo, R.~Sirvent, V.~Beltran, X.~Martorell, and J.~Labarta, ``Variable batched dgemm,'' in \emph{2018 26th Euromicro International Conference on Parallel, Distributed and Network-based Processing (PDP)}, 2018, pp. 363--367.

\bibitem{masliah2019algorithms}
I.~Masliah, A.~Abdelfattah, A.~Haidar, S.~Tomov, M.~Baboulin, J.~Falcou, and J.~Dongarra, ``Algorithms and optimization techniques for high-performance matrix-matrix multiplications of very small matrices,'' \emph{Parallel Computing}, vol.~81, pp. 1--21, 2019.

\bibitem{liu2014efficient}
W.~Liu and B.~Vinter, ``An efficient gpu general sparse matrix-matrix multiplication for irregular data,'' in \emph{IEEE 28th international parallel and distributed processing symposium}.\hskip 1em plus 0.5em minus 0.4em\relax IEEE, 2014, pp. 370--381.

\bibitem{nagasaka2018high}
Y.~Nagasaka, S.~Matsuoka, A.~Azad, and A.~Bulu{\c{c}}, ``High-performance sparse matrix-matrix products on intel knl and multicore architectures,'' in \emph{Proceedings of the 47th International Conference on Parallel Processing Companion}, 2018, pp. 1--10.

\bibitem{zhang2020sparch}
Z.~Zhang, H.~Wang, S.~Han, and W.~J. Dally, ``Sparch: Efficient architecture for sparse matrix multiplication,'' in \emph{2020 IEEE International Symposium on High Performance Computer Architecture (HPCA)}.\hskip 1em plus 0.5em minus 0.4em\relax IEEE, 2020, pp. 261--274.

\bibitem{ran2022cryptogcn}
R.~Ran, N.~Xu, W.~Wang, Q.~Gang, J.~Yin, and W.~Wen, ``Cryptogcn: Fast and scalable homomorphically encrypted graph convolutional network inference,'' \emph{arXiv preprint arXiv:2209.11904}, 2022.

\bibitem{patra2021aby2}
A.~Patra, T.~Schneider, A.~Suresh, and H.~Yalame, ``$\{$ABY2. 0$\}$: Improved $\{$Mixed-Protocol$\}$ secure $\{$Two-Party$\}$ computation,'' in \emph{30th USENIX Security Symposium (USENIX Security 21)}, 2021, pp. 2165--2182.

\bibitem{choi2019hybrid}
J.~I. Choi, D.~Tian, G.~Hernandez, C.~Patton, B.~Mood, T.~Shrimpton, K.~R. Butler, and P.~Traynor, ``A hybrid approach to secure function evaluation using sgx,'' in \emph{Proceedings of the 2019 ACM Asia Conference on Computer and Communications Security}, 2019, pp. 100--113.

\bibitem{husted2013gpu}
N.~Husted, S.~Myers, A.~Shelat, and P.~Grubbs, ``Gpu and cpu parallelization of honest-but-curious secure two-party computation,'' in \emph{Proceedings of the 29th Annual Computer Security Applications Conference}, 2013, pp. 169--178.

\bibitem{zhang2013picco}
Y.~Zhang, A.~Steele, and M.~Blanton, ``Picco: a general-purpose compiler for private distributed computation,'' in \emph{Proceedings of the 2013 ACM SIGSAC conference on Computer \& communications security}, 2013, pp. 813--826.

\bibitem{vasiljeva2017cloud}
T.~Vasiljeva, S.~Shaikhulina, and K.~Kreslins, ``Cloud computing: Business perspectives, benefits and challenges for small and medium enterprises (case of latvia),'' \emph{Procedia Engineering}, vol. 178, pp. 443--451, 2017.

\bibitem{ibarrondo2021pyfhel}
A.~Ibarrondo and A.~Viand, ``Pyfhel: Python for homomorphic encryption libraries,'' in \emph{Proceedings of the 9th on Workshop on Encrypted Computing \& Applied Homomorphic Cryptography}, 2021, pp. 11--16.

\bibitem{Huang2023TDSC}
Z.~Huang, C.~Hong, C.~Weng, W.-j. Lu, and H.~Qu, ``More efficient secure matrix multiplication for unbalanced recommender systems,'' \emph{IEEE Transactions on Dependable and Secure Computing}, vol.~20, no.~1, pp. 551--562, 2023.

\bibitem{huang2022secure}
H.~Huang and H.~Zong, ``Secure matrix multiplication based on fully homomorphic encryption,'' \emph{Journal of Supercomputing}, pp. 1--22, 2022.

\bibitem{jiang2018secure}
X.~Jiang, M.~Kim, K.~Lauter, and Y.~Song, ``Secure outsourced matrix computation and application to neural networks,'' in \emph{Proceedings of the 2018 ACM SIGSAC conference on computer and communications security}, 2018, pp. 1209--1222.

\bibitem{gupta2018oversketch}
V.~Gupta, S.~Wang, T.~Courtade, and K.~Ramchandran, ``Oversketch: Approximate matrix multiplication for the cloud,'' in \emph{2018 IEEE International Conference on Big Data (Big Data)}.\hskip 1em plus 0.5em minus 0.4em\relax IEEE, 2018, pp. 298--304.

\bibitem{dwork2006calibrating}
C.~Dwork, F.~McSherry, K.~Nissim, and A.~Smith, ``Calibrating noise to sensitivity in private data analysis,'' in \emph{Theory of cryptography conference}.\hskip 1em plus 0.5em minus 0.4em\relax Springer, 2006, pp. 265--284.

\bibitem{dwork2014algorithmic}
C.~Dwork, A.~Roth \emph{et~al.}, ``The algorithmic foundations of differential privacy,'' \emph{Foundations and Trends{\textregistered} in Theoretical Computer Science}, vol.~9, no. 3--4, pp. 211--407, 2014.

\bibitem{dwork2011firm}
C.~Dwork, ``A firm foundation for private data analysis,'' \emph{Communications of the ACM}, vol.~54, no.~1, pp. 86--95, 2011.

\bibitem{halevi2014algorithms}
S.~Halevi and V.~Shoup, ``Algorithms in helib,'' in \emph{Annual Cryptology Conference}.\hskip 1em plus 0.5em minus 0.4em\relax Springer, 2014, pp. 554--571.

\bibitem{smart2014fully}
N.~P. Smart and F.~Vercauteren, ``Fully homomorphic simd operations,'' \emph{Designs, codes and cryptography}, vol.~71, no.~1, pp. 57--81, 2014.

\bibitem{rathee2018faster}
D.~Rathee, P.~K. Mishra, and M.~Yasuda, ``Faster pca and linear regression through hypercubes in helib,'' in \emph{Proceedings of the 2018 Workshop on Privacy in the Electronic Society}, 2018, pp. 42--53.

\bibitem{yao1982protocols}
A.~C. Yao, ``Protocols for secure computations,'' in \emph{23rd annual symposium on foundations of computer science (sfcs 1982)}.\hskip 1em plus 0.5em minus 0.4em\relax IEEE, 1982, pp. 160--164.

\bibitem{lu2016using}
W.-j. Lu, S.~Kawasaki, and J.~Sakuma, ``Using fully homomorphic encryption for statistical analysis of categorical, ordinal and numerical data,'' \emph{Cryptology ePrint Archive}, 2016.

\bibitem{yasuda2015new}
M.~Yasuda, T.~Shimoyama, J.~Kogure, K.~Yokoyama, and T.~Koshiba, ``New packing method in somewhat homomorphic encryption and its applications,'' \emph{Security and Communication Networks}, vol.~8, no.~13, pp. 2194--2213, 2015.

\bibitem{naehrig2011can}
M.~Naehrig, K.~Lauter, and V.~Vaikuntanathan, ``Can homomorphic encryption be practical?'' in \emph{Proceedings of the 3rd ACM workshop on Cloud computing security workshop}, 2011, pp. 113--124.

\bibitem{rivest1978data}
R.~L. Rivest, L.~Adleman, M.~L. Dertouzos \emph{et~al.}, ``On data banks and privacy homomorphisms,'' \emph{Foundations of secure computation}, vol.~4, no.~11, pp. 169--180, 1978.

\bibitem{gentry2009fully}
C.~Gentry, ``Fully homomorphic encryption using ideal lattices,'' in \emph{Proceedings of the forty-first annual ACM symposium on Theory of computing}, 2009, pp. 169--178.

\bibitem{fan2012BFV}
J.~Fan and F.~Vercauteren, ``Somewhat practical fully homomorphic encryption,'' \emph{Cryptology ePrint Archive}, 2012.

\bibitem{Bfv12}
Z.~Brakerski, ``Fully homomorphic encryption without modulus switching from classical gapsvp,'' in \emph{Annual Cryptology Conference}.\hskip 1em plus 0.5em minus 0.4em\relax Springer, 2012, pp. 868--886.

\bibitem{brakerski2014BGV}
Z.~Brakerski, C.~Gentry, and V.~Vaikuntanathan, ``(leveled) fully homomorphic encryption without bootstrapping,'' \emph{ACM Transactions on Computation Theory (TOCT)}, vol.~6, no.~3, pp. 1--36, 2014.

\bibitem{cheon2017CKKS}
J.~H. Cheon, A.~Kim, M.~Kim, and Y.~Song, ``Homomorphic encryption for arithmetic of approximate numbers,'' in \emph{International conference on the theory and application of cryptology and information security}.\hskip 1em plus 0.5em minus 0.4em\relax Springer, 2017, pp. 409--437.

\bibitem{albrecht2021homomorphic}
M.~Albrecht, M.~Chase, H.~Chen, J.~Ding, S.~Goldwasser, S.~Gorbunov, S.~Halevi, J.~Hoffstein, K.~Laine, K.~Lauter \emph{et~al.}, ``Homomorphic encryption standard,'' in \emph{Protecting Privacy through Homomorphic Encryption}.\hskip 1em plus 0.5em minus 0.4em\relax Springer, 2021, pp. 31--62.

\bibitem{fheschemesbfv}
Inferati, ``Introduction to the bfv encryption scheme,'' \url{https://inferati.com/blog/fhe-schemes-bfv}, accessed Oct 4, 2022.

\bibitem{enwiki:1112117357}
\BIBentryALTinterwordspacing
{Wikipedia contributors}, ``Single instruction, multiple data --- {Wikipedia}{,} the free encyclopedia,'' 2022, [Online; accessed 4-October-2022]. [Online]. Available: \url{https://en.wikipedia.org/w/index.php?title=Single_instruction,_multiple_data&oldid=1112117357}
\BIBentrySTDinterwordspacing

\bibitem{gentry2011implementing10}
C.~Gentry and S.~Halevi, ``Implementing gentry’s fully-homomorphic encryption scheme,'' in \emph{Annual international conference on the theory and applications of cryptographic techniques}.\hskip 1em plus 0.5em minus 0.4em\relax Springer, 2011, pp. 129--148.

\bibitem{smart2010fully16}
N.~P. Smart and F.~Vercauteren, ``Fully homomorphic encryption with relatively small key and ciphertext sizes,'' in \emph{International Workshop on Public Key Cryptography}.\hskip 1em plus 0.5em minus 0.4em\relax Springer, 2010, pp. 420--443.

\bibitem{halevi2021bootstrapping}
S.~Halevi and V.~Shoup, ``Bootstrapping for helib,'' \emph{Journal of Cryptology}, vol.~34, no.~1, pp. 1--44, 2021.

\end{thebibliography}

\newpage 
\appendices

\section{The proof for the theorem~\ref{thrm:sigma-diag}$\sim$theorem~\ref{thrm:omg-diag}}
\label{apx:proof}

\begin{theorem*}~{\rm \textbf{\ref{thrm:sigma-diag}}}
Let $\sigma(\mathcal{A}) = \textbf{U}^\sigma \mathcal{A}$ for $\mathcal{A}$ with a dimension of $m\times l$. There are at most $2\cdot \min(m,l)-1$ non-zero diagonals in $\textbf{U}^\sigma$ no matter if the matrix is flattened with a column-major or row-major order. 
\end{theorem*}
\begin{proof}
When applying $\sigma$ transformation on matrix $\mathcal{A}_{m\times l}$ in \textbf{column-major} order, $\textbf{U}^\sigma$ is formulated in Equation~(\ref{eqn:Usigma}). Note that  $\textbf{U}^\sigma_{i+j\cdot m,h} = 1$ when $h= i+ [i+j]_l\cdot m$ and, for all elements of $\textbf{U}^\sigma_{i+j\cdot m,h}$ that belong to the same diagonal, we have $h-(i+j\cdot m)$ as a constant. 

Considering all the non-zero elements in $\textbf{U}^\sigma_{i+j\cdot m,h}$, we have
\begin{eqnarray}
h-(i+j\cdot m) &=& i+ [i+j]_l\cdot m - (i+j\cdot m) \nonumber \\ 
&=& i+(i+j-\left \lfloor \frac{i+j}{l} \right \rfloor \cdot l) \cdot m - (i+j\cdot m) \nonumber \\ 
&=&(i-\left \lfloor \frac{i+j}{l} \right \rfloor \cdot l)\cdot m. \nonumber    
\end{eqnarray}
Since $\left \lfloor \frac{i}{l} \right \rfloor + \left \lfloor \frac{j}{l} \right \rfloor 
\le 
\left \lfloor \frac{i+j}{l} \right \rfloor 
\le 
\left \lfloor \frac{i}{l} \right \rfloor + \left \lfloor \frac{j}{l} \right \rfloor +1$ and $0\le j < l$,  we have $\left \lfloor \frac{i}{l} \right \rfloor 
\le 
\left \lfloor \frac{i+j}{l} \right \rfloor 
\le 
\left \lfloor \frac{i}{l} \right \rfloor+1$.

Now consider two different scenarios: 1) $m<l$; 2) $m \geq l$. When  $m<l$, for each $i=\{1,2,...,m-1\}$, $h-(i+j\cdot m)$ can at most take two constant values since $\left \lfloor \frac{i}{l} \right \rfloor =0$ and $0 
\le 
\left \lfloor \frac{i+j}{l} \right \rfloor 
\le 
1$. When $i=0$, $h-(i+j\cdot m)$ can only be zero since $\left \lfloor \frac{i+j}{l} \right \rfloor =0$. Therefore, $\textbf{U}^\sigma_{i+j\cdot m,h}$ has at most $2m-1$ non-zero diagonals under this case. 

When $m \geq l$, we have 
\begin{eqnarray}
h-(i+j\cdot m) &=&(i-\left \lfloor \frac{i+j}{l} \right \rfloor \cdot l)\cdot m \nonumber \\
&=& (\left \lfloor \frac{i}{l} \right \rfloor \cdot l + p -\left \lfloor \frac{i+j}{l} \right \rfloor \cdot l)\cdot m, \nonumber
\end{eqnarray}
with $0 \leq p < l$.  Since $-1 \leq (\left \lfloor \frac{i}{l} \right \rfloor - \left \lfloor \frac{i+j}{l} \right \rfloor) 
\le 
0$, $\textbf{U}^\sigma_{i+j\cdot m,h}$ has at most $2l-1$ non-zero diagonals under this case. 

Therefore, in summary, there are at most $2\cdot \min(m,l)-1$ non-zero diagonals in $\textbf{U}^\sigma$ when the matrix is flattened with a column-major. Similar proof can be obtained when the matrix is flattened with the row-major order. 
\end{proof}

\begin{theorem*}~{\rm \textbf{\ref{thrm:tau-diag}}}
Let $\tau(\mathcal{B}) = \textbf{U}^\tau \mathcal{B}$ for $\mathcal{B}$ with a dimension of $l\times n$. There are at most $2\cdot \min(n,l)-1$ non-zero diagonals in $\textbf{U}^\tau$ no matter if the matrix is flattened with a column-major or row-major order. 
\end{theorem*}
\begin{proof}
When applying $\tau$ transformation on matrix $\mathcal{B}_{l\times n}$ in \textbf{column-major} order, $\textbf{U}^\tau$ is formulated in Equation~(\ref{eqn:Utau}). Note that  $\textbf{U}^\tau_{i+j\cdot l,h} = 1$ when $h= [i+j]_l + j\cdot l$ and, for all elements of $\textbf{U}^\tau_{i+j\cdot l,h}$ that belong to the same diagonal, we have $h-(i+j\cdot l)$ as a constant. 

Considering all the non-zero elements in $\textbf{U}^\tau_{i+j\cdot l,h}$, we have
\begin{eqnarray}
h-(i+j\cdot l) &=& [i+j]_l + j\cdot l - (i+j\cdot m) \nonumber \\ 
&=& i+j-\left \lfloor \frac{i+j}{l} \right \rfloor \cdot l+j\cdot l - (i+j\cdot m) \nonumber \\ 
&=&j-\left \lfloor \frac{i+j}{l} \right \rfloor\cdot l. \nonumber    
\end{eqnarray}
Since $\left \lfloor \frac{i}{l} \right \rfloor + \left \lfloor \frac{j}{l} \right \rfloor 
\le 
\left \lfloor \frac{i+j}{l} \right \rfloor
\le 
\left \lfloor \frac{i}{l} \right \rfloor + \left \lfloor \frac{j}{l} \right \rfloor +1$ and $0\le i < l$,  we have $\left \lfloor \frac{j}{l} \right \rfloor 
\le 
\left \lfloor \frac{i+j}{l} \right \rfloor 
\le 
\left \lfloor \frac{j}{l} \right \rfloor+1$.

Now consider two different scenarios: 1) $n<l$; 2) $n \geq l$. When  $n<l$, for each $j=\{1,2,...,n-1\}$, $h-(i+j\cdot l)$ can at most take two constant values since $\left \lfloor \frac{j}{l} \right \rfloor =0$ and $0 
\le 
\left \lfloor \frac{i+j}{l} \right \rfloor 
\le 
1$. When $i=0$, $h-(i+j\cdot l)$ can only be zero since $\left \lfloor \frac{i+j}{l} \right \rfloor =0$. Therefore, $\textbf{U}^\tau_{i+j\cdot l,h}$ has at most $2n-1$ non-zero diagonals under this case. 

When $n \geq l$, we have 
\begin{eqnarray}
h-(i+j\cdot l) &=&j-\left \lfloor \frac{i+j}{l} \right \rfloor\cdot l \nonumber \\
&=& \left \lfloor \frac{j}{l} \right \rfloor \cdot l + p -\left \lfloor \frac{i+j}{l} \right \rfloor \cdot l, \nonumber
\end{eqnarray}
with $0 \leq p < l$.  Since $-1 \leq (\left \lfloor \frac{j}{l} \right \rfloor - \left \lfloor \frac{i+j}{l} \right \rfloor) 
\le 
0$, $\textbf{U}^\tau_{i+j\cdot l,h}$ has at most $2l-1$ non-zero diagonals under this case. 

Therefore, in summary, there are at most $2\cdot \min(n,l)-1$ non-zero diagonals in $\textbf{U}^\tau$ when the matrix is flattened with a column-major. Similar proof can be obtained when the matrix is flattened with the row-major order. 
\end{proof}

\begin{theorem*}~{\rm \textbf{\ref{thrm:eps-diag}}}
Let $\epsilon^k_{m\times n}(\mathcal{A}) =\textbf{U}^{\epsilon^k_{m\times n}} \mathcal{A}$ be the linear transformation $\epsilon_{m\times n}: \mathcal{R}_{m\times l} \rightarrow \mathcal{R}_{m\times n}$ with matrix $\mathcal{A}$ having a dimension of $m\times l$. 
There are at most $\left \lfloor \frac{n}{l} \right \rfloor +1$ non-zero diagonal vectors in $\textbf{U}^{\epsilon^k_{m\times n}}$ when the matrix is flattened with the \textbf{column-major} order; 
There are at most $(\left \lfloor \frac{n}{l} \right \rfloor +2)\cdot m$ non-zero diagonal vectors in $\textbf{U}^{\epsilon^k_{m\times n}}$ when matrix $\mathcal{A}$ is flattened with the \textbf{row-major} order. 
Specifically, when $n=l$, there are no more than 2 non-zero diagonals in $\textbf{U}^{\epsilon^k_{m\times n}}$, no matter if the matrix is flattened in column-major or row-major order.
\end{theorem*}
\begin{proof}
When applying $\epsilon$ transformation on matrix $\mathcal{A}_{m\times l}$ in \textbf{column-major} order, $\textbf{U}^\epsilon$ is formulated in Equation~(\ref{eqn:Uepsc}). Note that  $\textbf{U}^{\epsilon^k_{m\times n}}_{i,j} = 1$ when $j=[k\cdot m+i]_{m\cdot l}$ and, for all elements of $\textbf{U}^{\epsilon^k_{m\times n}}_{i,j}$ that belong to the same diagonal, we have $j-i$ as a constant. 

Considering all the non-zero elements in $\textbf{U}^{\epsilon^k_{m\times n}}_{i,j}$, we have
\begin{eqnarray}
j-i &=& [k\cdot m+i]_{m\cdot l} - i \nonumber \\ 
&=& k\cdot m+i-\left \lfloor \frac{k\cdot m+i}{m\cdot l}\right \rfloor \cdot m\cdot l - i \nonumber \\ 
&=&k\cdot m-\left \lfloor \frac{k\cdot m+i}{m\cdot l}\right \rfloor \cdot m\cdot l \nonumber    
\end{eqnarray}

Since $\max(k)=l-1$ and $\max(i)=m\cdot n-1$, we have 
\begin{eqnarray}
\max(\frac{k\cdot m+i}{m\cdot l})
&<&\frac{l-1+n}{l} \nonumber \\ 
&\le& \left \lfloor \frac{l-1}{l} \right \rfloor + \left \lfloor \frac{n}{l} \right \rfloor +1 \nonumber \\ 
&=&\left \lfloor \frac{n}{l} \right \rfloor +1 \nonumber 
\end{eqnarray}

Therefore, we get $\left \lfloor \frac{k\cdot m+i}{m\cdot l}\right \rfloor \in \{0,1,...,\left \lfloor \frac{n}{l} \right \rfloor\}$. Then, $j-i=k\cdot m-\left \lfloor \frac{k\cdot m+i}{m\cdot l}\right \rfloor\cdot m\cdot l$. $k$, $m$ and $l$ are all constant number for one transformation. 
The set $\{0,1,...,\left \lfloor \frac{n}{l} \right \rfloor\}$ is of size $\left \lfloor \frac{n}{l} \right \rfloor +1$. In summary, $\textbf{U}^{\epsilon^k_{m\times n}}$ has at most $\left \lfloor \frac{n}{l} \right \rfloor +1$ constant values when $\mathcal{A}_{m\times l}$ in \textbf{column-major}.

Special circumstances is when $n=l$, $\left \lfloor \frac{n}{l} \right \rfloor=1$. Therefore, $\left \lfloor \frac{n}{l} \right \rfloor+1=2$ and this means $\textbf{U}^{\epsilon^k_{m\times n}}$ has only \textbf{2} non-zero diagonals when $n=l$..

When applying $\epsilon$ transformation on matrix $\mathcal{A}_{m\times l}$ in \textbf{row-major} order, we can formulate permutation matrix according to formula~(\ref{eqn:Uomgc}), but apply on $\mathcal{A}_{l\times m}$ instead of $\mathcal{A}_{l\times n}$. Note that  $\textbf{U}^{\epsilon^k_{m\times n}}_{i,j} = 1$ when $j=[k+[i]_n]_{l}+\left \lfloor i/n \right \rfloor \cdot l$ and, for all elements of $\textbf{U}^{\epsilon^k_{m\times n}}_{i,j}$ that belong to the same diagonal, we have $j-i$ as a constant. 

Considering all the non-zero elements in $\textbf{U}^{\epsilon^k_{m\times n}}_{i,j}$, we have
\begin{eqnarray}
j&=& k+[i]_n-\left \lfloor \frac{k+[i]_n}{l}\right \rfloor \cdot l+\left \lfloor \frac{i}{n} \right \rfloor\cdot l\nonumber \\ 
&=&k+[i]_n+(\left \lfloor \frac{i}{n} \right \rfloor - \left \lfloor \frac{k+[i]_n}{l}\right \rfloor)\cdot l \nonumber 
\end{eqnarray}

Since $i\in[0,mn)$, we split $i$ to $m$ circumstances that $i\in[pn,(p+1)n)$ where $p =\{0,1,2,...,m-1\}$. For for each circumstance that $i\in[pn,(p+1)n)$, we have
\begin{equation}
j= k+i-pn+(p - \left \lfloor \frac{k+[i]_n}{l}\right \rfloor)\cdot l \nonumber 
\end{equation}
and
\begin{equation}
j-i= k-pn+(p - \left \lfloor \frac{k+[i]_n}{l}\right \rfloor)\cdot l \nonumber 
\end{equation}
Note that we have
\begin{equation}
    \left \lfloor \frac{[pn]_n}{l}\right \rfloor
\le
\left \lfloor \frac{k+[i]_n}{l}\right \rfloor
<
\left \lfloor \frac{[pn]_n}{l}\right \rfloor +
\left \lfloor \frac{n}{l}\right \rfloor +1+1
\nonumber
\end{equation}
which has $2 + \left \lfloor \frac{n}{l} \right \rfloor$ constant values. And this means $j-i$ , which represents the number of non-zero diagonals in $\textbf{U}^{\epsilon^k_{m\times n}}$, has $(2 + \left \lfloor \frac{n}{l} \right \rfloor)\cdot m$ in total when $\mathcal{A}_{m\times l}$ in \textbf{row-major} because there are $m$ circumstances.

Special circumstances is when $n=l$, $j-i \in \{0,1\}$. The reason is that, since
\begin{equation}
    \left \lfloor \frac{k}{l}\right \rfloor + \left \lfloor \frac{[i]_n}{l}\right \rfloor \le \left \lfloor \frac{k+[i]_n}{l}\right \rfloor \le \left \lfloor \frac{k}{l}\right \rfloor + \left \lfloor \frac{[i]_n}{l}\right \rfloor +1 \nonumber
\end{equation}
and we also have $k<l$ and $ {[i]_n}<{l}$, thus
\begin{equation}
    0 \le \left \lfloor \frac{k+[i]_n}{l}\right \rfloor \le 1 \nonumber
\end{equation}
On the other hand, we have
\begin{equation}
j-i= k - \left \lfloor \frac{k+[i]_n}{l}\right \rfloor\cdot l \nonumber 
\end{equation}
for each $i\in[pn,(p+1)n)$. $j-i$ has the same constant value in each $i\in[pn,(p+1)n)$ and this means $\textbf{U}^{\epsilon^k_{m\times n}}$ has only \textbf{2} non-zero diagonals when $n=l$.
\end{proof}

\begin{theorem*}~{\rm \textbf{\ref{thrm:omg-diag}}}
Let $\omega^k_{m\times n}(\mathcal{B}) = \textbf{U}^{\omega^k_{m\times n}} \mathcal{B}$ be the linear transformation $\omega_{m\times n}: \mathcal{R}_{l\times n} \rightarrow \mathcal{R}_{m\times n}$ with matrix $\mathcal{B}$ having a dimension of $l\times n$. 
There are at most $(\left \lfloor \frac{m}{l} \right \rfloor+2)\cdot n$ non-zero diagonal vectors in $\textbf{U}^{\omega^k_{m\times n}}$ when the matrix is flattened with \textbf{column-major} order; There are
at most $\left \lfloor \frac{m}{l} \right \rfloor +1$ non-zero diagonal vectors in $\textbf{U}^{\omega^k_{m\times n}}$ when matrix $\mathcal{B}$ is flattened with \textbf{row-major} order. 
Specifically, when $m=l$, there are no more than 2 non-zero diagonals in $\textbf{U}^{\omega^k_{m\times n}}$, no matter if the matrix is flattened in column-major or row-major order.
 \end{theorem*}
\begin{proof}
When applying $\omega$ transformation on matrix $\mathcal{B}_{l\times n}$ in \textbf{column-major} order, $\textbf{U}^\omega$ is formulated in Equation~(\ref{eqn:Uomgc}). Note that  $\textbf{U}^{\omega^k_{m\times n}}_{i,j} = 1$ when $j=[k+[i]_m]_{l}+\left \lfloor i/m \right \rfloor \cdot l$ and, for all elements of $\textbf{U}^{\omega^k_{m\times n}}_{i,j}$ that belong to the same diagonal, we have $j-i$ as a constant. 

Considering all the non-zero elements in $\textbf{U}^{\omega^k_{m\times n}}_{i,j}$, we have
\begin{eqnarray}
j&=& k+[i]_m-\left \lfloor \frac{k+[i]_m}{l}\right \rfloor \cdot l+\left \lfloor \frac{i}{m} \right \rfloor\cdot l\nonumber \\ 
&=&k+[i]_m+(\left \lfloor \frac{i}{m} \right \rfloor - \left \lfloor \frac{k+[i]_m}{l}\right \rfloor)\cdot l \nonumber 
\end{eqnarray}

Since $i\in[0,mn)$, we split $i$ to $n$ circumstances that $i\in[pm,(p+1)m)$ where $p =\{0,1,2,...,n-1\}$. For each circumstance that $i\in[pm,(p+1)m)$, we have
\begin{equation}
j= k+i-pm+(p - \left \lfloor \frac{k+[i]_m}{l}\right \rfloor)\cdot l \nonumber 
\end{equation}
and
\begin{equation}
j-i= k-pm+(p - \left \lfloor \frac{k+[i]_m}{l}\right \rfloor)\cdot l \nonumber 
\end{equation}
Note that we have
\begin{equation}
    \left \lfloor \frac{[pm]_m}{l}\right \rfloor
\le
\left \lfloor \frac{k+[i]_m}{l}\right \rfloor
<
\left \lfloor \frac{[pm]_m}{l}\right \rfloor +
\left \lfloor \frac{m}{l}\right \rfloor +1+1
\nonumber
\end{equation}
which has $2 + \left \lfloor \frac{m}{l} \right \rfloor$ constant values. And this means $j-i$ , which represents the number of non-zero diagonals in $\textbf{U}^{\omega^k_{m\times n}}$, has $(2 + \left \lfloor \frac{m}{l} \right \rfloor)\cdot n$ in total when $\mathcal{B}_{m\times l}$ in \textbf{row-major} because there are $n$ circumstances.

Special circumstances is when $m=l$, $j-i \in \{0,1\}$. The reason is that, since
\begin{equation}
    \left \lfloor \frac{k}{l}\right \rfloor + \left \lfloor \frac{[i]_l}{l}\right \rfloor \le \left \lfloor \frac{k+[i]_l}{l}\right \rfloor \le \left \lfloor \frac{k}{l}\right \rfloor + \left \lfloor \frac{[i]_l}{l}\right \rfloor +1 \nonumber
\end{equation}
and we also have $k<l$ and $ {[i]_l}<{l}$, thus
\begin{equation}
    0 \le \left \lfloor \frac{k+[i]_l}{l}\right \rfloor \le 1 \nonumber
\end{equation}
On the other hand, we have
\begin{equation}
j-i= k - \left \lfloor \frac{k+[i]_l}{l}\right \rfloor\cdot l \nonumber 
\end{equation}
for each $i\in[pm,(p+1)m)$. $j-i$ has the same constant value in each $i\in[pm,(p+1)m)$ and this means $\textbf{U}^{\omega^k_{m\times n}}$ has only \textbf{2} non-zero diagonals when $m=l$.

When applying $\omega$ transformation on matrix $\mathcal{B}_{l\times n}$ in \textbf{row-major} order, we can formulate permutation matrix according to formula~(\ref{eqn:Uepsc}), but apply on $\mathcal{B}_{n\times l}$ instead of $\mathcal{B}_{m\times l}$. Note that  $\textbf{U}^{\omega^k_{m\times n}}_{i,j} = 1$ when $j=[k\cdot n+i]_{n\cdot l}$ and, for all elements of $\textbf{U}^{\omega^k_{m\times n}}_{i,j}$ that belong to the same diagonal, we have $j-i$ as a constant. 

Considering all the non-zero elements in $\textbf{U}^{\omega^k_{m\times n}}_{i,j}$, we have
\begin{eqnarray}
j-i &=& [k\cdot n+i]_{n\cdot l} - i \nonumber \\ 
&=& k\cdot n+i-\left \lfloor \frac{k\cdot n+i}{n\cdot l}\right \rfloor \cdot n\cdot l - i \nonumber \\ 
&=&k\cdot n-\left \lfloor \frac{k\cdot n+i}{n\cdot l}\right \rfloor \cdot n\cdot l \nonumber    
\end{eqnarray}

Since $\max(k)=l-1$ and $\max(i)=m\cdot n-1$, we have 
\begin{eqnarray}
\max(\frac{k\cdot n+i}{n\cdot l})
&<&\frac{l-1+m}{l} \nonumber \\ 
&\le& \left \lfloor \frac{l-1}{l} \right \rfloor + \left \lfloor \frac{m}{l} \right \rfloor +1 \nonumber \\ 
&=&\left \lfloor \frac{m}{l} \right \rfloor +1 \nonumber 
\end{eqnarray}

Therefore, we get $\left \lfloor \frac{k\cdot n+i}{n\cdot l}\right \rfloor \in \{0,1,...,\left \lfloor \frac{m}{l} \right \rfloor\}$. Then, $j-i=k\cdot n-\left \lfloor \frac{k\cdot n+i}{n\cdot l}\right \rfloor\cdot n\cdot l$. Here, $k$, $n$ and $l$ are all constant number for one transformation. 
The set $\{0,1,...,\left \lfloor \frac{m}{l} \right \rfloor\}$ is of size $\left \lfloor \frac{m}{l} \right \rfloor +1$. In summary, $\textbf{U}^{\omega^k_{m\times n}}$ has at most $\left \lfloor \frac{m}{l} \right \rfloor +1$ constant values when $\mathcal{B}_{m\times l}$ in \textbf{row-major}.

Special circumstances is when $m=l$, $\left \lfloor \frac{m}{l} \right \rfloor=1$. Therefore, $\left \lfloor \frac{m}{l} \right \rfloor+1=2$ and this means $\textbf{U}^{\omega^k_{m\times n}}$ has only \textbf{2} non-zero diagonals when $m=l$..
\end{proof}

\begin{theorem*}~{\rm \textbf{\ref{thrm:dupA}}}
Let $\mathcal{A}_{m\times l}$ and $\mathcal{B}_{l\times n}$ with $m < l$, and let $\bar{A}$ be matrix expanded with $t=\left \lceil \frac{l}{m} \right \rceil$ copies of $\mathcal{A}$ vertically, i.e., $\bar{\mathcal{A}} = \{\bar{A_0}; \bar{A_1}; ...; \bar{A}_{(t-1)}\}^T$ 
with $\bar{A_0}=\bar{A_1}=...= \bar{A}_{(t-1)}=\mathcal{A}_{m\times l}$. Then
\begin{itemize}
\item $\epsilon^k_{tm\times n} ( \sigma(\bar{\mathcal{A}}))\odot \omega^k_{tm\times n} ( \tau(\mathcal{B}))$ contains $t$ items of $\epsilon^p_{m\times n} ( \sigma(\mathcal{A}))\odot \omega^p_{m\times n} ( \tau(\mathcal{B}))$, with $p \in \{[k]_l,[k+m]_l,..., [k+(t-1)m]_l\}$.
\item $\epsilon^k_{tm\times n} ( \sigma(\bar{\mathcal{A}}))\odot \omega^k_{tm\times n} ( \tau(\mathcal{B}))$, $k=0,1,...,(m-1)$ contains all items of $\epsilon^p_{m\times n} (\sigma(\mathcal{A}))\odot \omega^p_{m\times n} (\tau(\mathcal{B}))$, with $p\in \{0, 1, ..., (l-1)\}$.
\end{itemize}
\end{theorem*}
\begin{proof}
Consider a sub matrix of $(\epsilon^k_{tm\times n} \circ \sigma({\bar{\mathcal{A}})})$ with dimension of $m\times n$, i.e., 
$ (\epsilon^k_{tm\times n} \circ \sigma({\bar{\mathcal{A}})})_{hm+i,j}$, where 
$0 \leq i < m, 0 \leq j < n$. $h$ is a constant with $0 \leq h < t$. Based on equation (\ref{def:sigma}) and (\ref{def:eps}), we have 
\begin{eqnarray}
    (\epsilon^k_{tm\times n} \circ \sigma({\bar{\mathcal{A}})})_{hm+i,j} &=&     \sigma({\bar{\mathcal{A}}})_{hm+i,[j+k]_l} \nonumber \\
    & = & {\bar{\mathcal{A}}}_{hm+i,[hm+i+j+k]_l} \nonumber \\
    &=& \mathcal{A}_{i,[hm+i+j+k]_l} \label{eqn:prf5-1}
\end{eqnarray}

On the other hand, let $p = [k+hm]_l$, for $0 \leq i < m, 0 \leq j < n$, we have 
\begin{eqnarray}
 (\epsilon^p_{m\times n} \circ \sigma({\mathcal{A}}))_{i,j}
 &=& \sigma({\mathcal{A})_{i,[j+p]_l}} \nonumber\\
 &=& \mathcal{A}_{i,[i+j+k+hm]_l}
\end{eqnarray}

Similarly, consider the sub matrix of $(\omega^k_{tm\times n} \circ \tau({\mathcal{B}}))$ with dimension of $m\times n$, i.e., 
$(\omega^k_{tm\times n} \circ \tau({\mathcal{B}}))_{hm+i,j}$, with $0 \leq i < m, 0 \leq j < n$. Based on equation (\ref{def:tau}) and (\ref{def:omega}), we have 
\begin{eqnarray}
    (\omega^k_{tm\times n} \circ \tau({\mathcal{B}}))_{hm+i,j}
    &=& \tau({\mathcal{B}})_{[hm+i+k]_l,j} \nonumber \\
   &=&{\mathcal{B}}_{[hm+i+j+k]_l,j}
\end{eqnarray}
If we let $p = [k+hm]_l$, for $0 \leq i < m, 0 \leq j < n$, and $0 \leq h < t$, we have 
\begin{eqnarray}
 \omega^p_{m\times n} \circ \tau({\mathcal{B}})_{i,j}
 &=& \tau({\mathcal{B}})_{[i+p]_l,j} \nonumber \\
 &=& {\mathcal{B}}_{[i+k+hm+j]_l,j} \label{eqn:prf5-4}
\end{eqnarray}

Since $0 \leq h < t$, there are total $t$ sub matrices in $\epsilon^k_{tm\times n} ( \sigma(\bar{\mathcal{A}}))$ and $\omega^k_{tm\times n}(\tau(\mathcal{B}))$, the conclusion for the first part of the theorem follows naturally from equation (\ref{eqn:prf5-1}) to (\ref{eqn:prf5-4}).

To prove the second part of the theorem, we only need to note that since $t=\left \lceil \frac{l}{m} \right \rceil$, we have $tm \ge l$. Therefore, for any $p\in \{0, 1, ..., (l-1)\}$, we must be able to find at least one set of $k$ and $h$, with $0 \leq k < m$, $0 \leq h <t$, and $p = [k+hm]_l$. Together with equation (\ref{eqn:prf5-1}) to (\ref{eqn:prf5-4}), we thus prove the theorem.

\end{proof}

\begin{theorem*}~{\rm \textbf{\ref{thrm:dupB}}}
Let $\mathcal{A}_{m\times l}$ and $\mathcal{B}_{l\times n}$ with $n < l$, and let $\bar{\mathcal{B}}$ be matrix expanded with $t=\left \lceil \frac{l}{n} \right \rceil$ copies of $\mathcal{B}$ horizontally, i.e., $\bar{\mathcal{B}} = \{\mathcal{B}; \mathcal{B}; ...; \mathcal{B}\}$. Then 

\begin{itemize}
\item {$\epsilon^k_{m\times tn}( \sigma(\mathcal{A}))\odot \omega^k_{m\times tn}( \tau(\bar{\mathcal{B}}))$ contains $t$ items of $\epsilon^p_{m\times n}( \sigma(\mathcal{A}))\odot \omega^p_{m\times n}( \tau(\mathcal{B}))$, with $p=[k]_l, [k+n]_l,..., [k+(t-1)n]_l$;}

\item {$\epsilon^k_{m\times tn}( \sigma(\mathcal{A}))\odot \omega^k_{m\times tn}( \tau(\bar{\mathcal{B}})$, $k=0,1,...,(n-1)$ contains all items of $\epsilon^p_{m\times n}( \sigma(\mathcal{A}))\odot \omega^p_{m\times n}( \tau(\mathcal{B}))$, with $p=0, 1, ..., (l-1)$.}
\end{itemize}
\end{theorem*}
\begin{proof}
Consider a sub matrix of $(\epsilon^k_{m\times tn} \circ \sigma({\mathcal{A})})$ with dimension of $m\times n$, i.e., 
$ (\epsilon^k_{m\times tn} \circ \sigma({\mathcal{A})})_{i,hn+j}$, where 
$0 \leq i < m, 0 \leq j < n$. $h$ is a constant with $0 \leq h < t$. Based on equation (\ref{def:sigma}) and (\ref{def:eps}), we have 
\begin{eqnarray}
    (\epsilon^k_{m\times tn} \circ \sigma({\mathcal{A})})_{i,hn+j} &=&     \sigma({\mathcal{A}})_{i,[hn+j+k]_l} \nonumber \\
    & = & {\mathcal{A}}_{i,[i+hn+j+k]_l} \label{eqn:prf6-1}
\end{eqnarray}

On the other hand, let $p = [k+hn]_l$, for $0 \leq i < m, 0 \leq j < n$, we have 
\begin{eqnarray}
 (\epsilon^p_{m\times n} \circ \sigma({\mathcal{A}}))_{i,j}
 &=& \sigma({\mathcal{A})_{i,[j+p]_l}} \nonumber\\
 &=& \mathcal{A}_{i,[i+j+k+hn]_l}.
\end{eqnarray}

Similarly, consider the sub matrix of $(\omega^k_{m\times tn} \circ \tau(\bar{\mathcal{B}}))$ with dimension of $m\times n$, i.e., 
$(\omega^k_{m\times tn} \circ \tau(\bar{\mathcal{B}}))_{i,hn+j}$, with $0 \leq i < m, 0 \leq j < n$. Based on equation (\ref{def:tau}) and (\ref{def:omega}), we have 
\begin{eqnarray}
    (\omega^k_{m\times tn} \circ \tau(\bar{\mathcal{B}}))_{i,hn+j}
    &=& \tau(\bar{\mathcal{B}})_{[i+k]_l,hn+j} \nonumber \\
   &=&\bar{\mathcal{B}}_{[hn+i+j+k]_l,hn+j} \nonumber \\
   &=&\mathcal{B}_{[hn+i+j+k]_l,j}
\end{eqnarray}
If we let $p = [k+hn]_l$, for $0 \leq i < m, 0 \leq j < n$, and $0 \leq h < t$, we have 
\begin{eqnarray}
 \omega^p_{m\times n} \circ \tau({\mathcal{B}})_{i,j}
 &=& \tau({\mathcal{B}})_{[i+p]_l,j} \nonumber \\
 &=& {\mathcal{B}}_{[i+k+hn+j]_l,j} \label{eqn:prf6-4}
\end{eqnarray}

Since $0 \leq h < t$, there are total $t$ sub matrices in $\epsilon^k_{m\times tn} ( \sigma(\mathcal{A}))$ and $\omega^k_{m\times tn}(\tau(\mathcal{B}))$, the conclusion for the first part of the theorem follows naturally from equation (\ref{eqn:prf6-1}) to (\ref{eqn:prf6-4}).

To prove the second part of the theorem, we only need to note that since $t=\left \lceil \frac{l}{n} \right \rceil$, we have $tn \ge l$. Therefore, for any $p\in \{0, 1, ..., (l-1)\}$, we must be able to find at least one set of $k$ and $h$, with $0 \leq k < m$, $0 \leq h <t$, and $p = [k+hn]_l$. Together with equation (\ref{eqn:prf6-1}) to (\ref{eqn:prf6-4}), we thus prove the theorem.

\end{proof}

\section{meaning of symbolize}
\begin{table}[htbp]
  \centering
  \caption{Meaning of symbolize}
    \begin{tabular}{c|l}
    \multicolumn{1}{c|}{\textbf{Symbolize}} & \multicolumn{1}{c}{\textbf{Meaning}}\\
    \hline
    $\mathcal{A}$ & left matrix for matrix multiplicaiton \\
    \hline
    $\mathcal{B}$ & right matrix for matrix multiplicaiton \\
    \hline
    $m$  & the number of row of matrix $\mathcal{A}$ \\
    \hline
    \multirow{2}{*}{$l$} & the number of column of matrix $\mathcal{A}$ \\
\cline{2-2}         & \multicolumn{1}{p{20.5em}}{the number of row of matrix $\mathcal{B}$} \\
    \hline
    $n$  & the number of column of matrix $\mathcal{B}$ \\
    \hline
    $\sigma$ & The transformation that permute each row  \\
    \hline
    $\tau$ & The transformation that permute each column  \\
    \hline
    $\epsilon$ & The transformation that permute mutiple columns \\
    \hline
    $\omega$ & The transformation that permute mutiple rows \\
    \hline
    $ct$ & the prefix  of ciphertext \\
    \hline
    $\textbf{U}$ & permutation matrix \\
    \hline
    $\odot$ & elemenwise multiplication\\
    \end{tabular}%
  \label{tab:addlabel}%
\end{table}%

\end{document}